\def\anon{0}
\documentclass[11pt]{article}
\ifnum\anon=1  \fi 

\usepackage{fullpage}
\usepackage{amssymb}
\usepackage{amsmath}
\usepackage{amsthm}
\usepackage{multirow}
\usepackage{enumerate}
\usepackage{graphicx}
\usepackage{mathrsfs}
\usepackage[utf8]{inputenc}
\usepackage{mdframed}
\usepackage{thmtools, thm-restate}

\usepackage[hidelinks,hypertexnames=false]{hyperref}
\usepackage[nameinlink]{cleveref}
\usepackage{url}            
\usepackage{footnotebackref}

\usepackage{amsfonts}
\usepackage[T1]{fontenc}
\usepackage{bm}
\usepackage{changepage}
\usepackage[dvipsnames,svgnames,table]{xcolor}
\usepackage{subcaption}
\usepackage{enumitem}
\usepackage{algpseudocode}

\usepackage{lineno}

\usepackage{bbm}

\newtheorem{theorem}{Theorem}[section]

\newtheorem{lemma}[theorem]{Lemma}
\newtheorem{observation}[theorem]{Observation}
\newtheorem{proposition}[theorem]{Proposition}
\newtheorem{definition}[theorem]{Definition}
\newtheorem{claim}[theorem]{Claim}
\newtheorem{fact}[theorem]{Fact}
\newtheorem{remark}[theorem]{Remark}

\newtheorem{question}{Question}

\newcommand{\change}[1]{#1}

\newcommand{\paren}[1]{\left(#1\right)}

\newcommand{\brac}[1]{\left[#1\right]}

\newcommand{\set}[1]{\left\{#1\right\}}

\newcommand{\veca}{{\mathbf{a}}}

\newcommand{\vecx}{{\mathbf{x}}}
\newcommand{\vecy}{{\mathbf{y}}}
\newcommand{\vecz}{{\mathbf{z}}}

\newcommand{\bool}{\{0,1\}}
\newcommand{\bs}[1][n]{\{0,1\}^{#1}_{#1/2}}
\newcommand{\bss}[1][n]{\{0,1\}^{2#1}_{#1}}
\newcommand{\bg}[1][n/2]{\{0,1\}^{#1}}
\newcommand{\pds}{\cP(D,d)}

\newcommand{\bx}{\mathbf{x}}
\newcommand{\by}{\mathbf{y}}

\newcommand{\cTgrid}{\mathcal{T}_{\textnormal{grid}}}
\newcommand{\cTslice}{\mathcal{T}_{\textnormal{slice}}}
\renewcommand{\ge}{\geqslant}
\renewcommand{\geq}{\geqslant}
\renewcommand{\le}{\leqslant}
\renewcommand{\leq}{\leqslant}
\renewcommand{\epsilon}{\varepsilon}

\newcommand{\F}{\mathbb{F}}

\newcommand{\R}{\mathbb{R}}

\newcommand{\Z}{\mathbb{Z}}

\newcommand{\cB}{\mathcal{B}}

\newcommand{\cD}{\mathcal{D}}

\newcommand{\cF}{\mathcal{F}}
\newcommand{\cG}{\mathcal{G}}

\newcommand{\cP}{\mathcal{P}}

\newcommand{\cT}{\mathcal{T}}

\DeclareMathOperator*{\E}{\mathbb{E}}

\newcommand{\poly}{\textnormal{poly}}

\usepackage[linesnumbered,ruled,vlined]{algorithm2e} 
\SetKwComment{Comment}{/* }{ */}

\SetCommentSty{mycommfont}

\definecolor{thmcolor}{RGB}{235, 235, 235}
\definecolor{citecolor}{RGB}{1, 210, 56}
\definecolor{lemmacolor}{RGB}{130, 169, 252}
\usepackage[most]{tcolorbox}
\newtcolorbox{algobox}{colback=lightgray!5!white,colframe=lightgray!75!black}
\newtcolorbox{thmbox}{colback=thmcolor!5!white,colframe=black!75!black}
\newtcolorbox{lemmabox}{colback=lemmacolor!5!white,colframe=blue!75!blue}

\hypersetup{
	colorlinks,
	linkcolor={blue},
	citecolor={citecolor},
	urlcolor={blue}
}

\usepackage[
	backend=biber,
	style=alphabetic,
	sorting=nyt,
	backref=true,
	maxcitenames = 8,
	mincitenames = 5,
	maxalphanames = 8,
	minalphanames = 5,
	maxnames = 10,
	minnames = 5
]{biblatex}

\begin{document}
\title{Low-Degree Testing Over Boolean Slices}

\if\anon0{\author{Prashanth Amireddy\thanks{School of Engineering and Applied Sciences, Harvard University, Cambridge, Massachusetts, USA. Supported in part by Madhu Sudan's Simons Investigator Award and NSF Award CCF 2152413 and Salil Vadhan's Simons Investigator Award. Part of this work was done during a visit to the University of Copenhagen supported by the European Research Council (ERC) under grant agreement no. 101125652 (ALBA). Email: \texttt{pamireddy@g.harvard.edu}.} \and
		Amik Raj Behera\thanks{Department of Computer Science, University of Copenhagen, Denmark. Supported by Srikanth Srinivasan's start-up grant from the University of Copenhagen. Email: \texttt{ambe@di.ku.dk}. } \and
		Srikanth Srinivasan \thanks{Department of Computer Science, University of Copenhagen, Denmark. Supported by the European Research Council (ERC) under grant agreement no. 101125652 (ALBA). Email: \texttt{srsr@di.ku.dk}. } \and
		Madhu Sudan\thanks{School of Engineering and Applied Sciences, Harvard University, Cambridge, Massachusetts, USA. Supported in part by a Simons Investigator Award, NSF Award CCF 2152413 and AFOSR award FA9550-25-1-0112. Email: \texttt{madhu@cs.harvard.edu}.} \and Sophus Valentin Willumsgaard \thanks{Department of Computer Science, University of Copenhagen, Denmark. Supported by the European Research Council (ERC) under grant agreement no. 101125652 (ALBA). Email: \texttt{sophus.willumsgaard@di.ku.dk}. }   }}\else{
}\fi
\maketitle

\begin{abstract}
	We study low-degree testing for group-valued functions over a Boolean slice. Specifically given a degree parameter $d$ and oracle access to a function
	$f:\{0,1\}^n_{n/2}\to G$ where $\{0,1\}^n_k$ denotes the set of vectors in $\{0,1\}^n$ of Hamming weight $k$ and $G$ is an Abelian group (not necessarily finite), the low-degree testing problem asks us to distinguish the case where $f$ is a polynomial of degree at most $d$ (with coefficients from $G$) or is $\epsilon$-far from the set of all such polynomials. Classical works in this area considered functions with domain $\F_q^n$ and range $\F_q$. More recent works have considered the setting where the domain is the Boolean cube [Bafna, Srinivasan, Sudan (Random Structures and Algorithms 2020), Amireddy, Srinivasan, Sudan (RANDOM 2023)], or when the domain is the slice (i.e., $\{0,1\}^n_{k}$) and the range is $\F_2$ [David, Dinur, Goldenberg, Kindler and Shinkar (SIAM  Journal on Computing 2017), Kalai, Lifshitz, Minzer and Ziegler (FOCS 2024)]. Each of the changes introduces new challenges in designing and analyzing low-degree tests and this happens again in our setting with domain being a slice and range is general. Indeed the previous methods fail even when the domain is a Boolean slice and the range is $\F_3$.\\

	Our main theorem gives a test that makes $O_d(1)$ (specifically $\exp(d^{O(1)})$) queries to $f$ and accepts degree-$d$ functions while rejecting functions that are $\epsilon$-far with probability $\Omega(\varepsilon)$.
	The central proof idea is to reduce this low-degree testing problem to the problem of low-degree testing on the cube. Specifically we show how to randomly embed the $n/2$-dimensional cube $\{0,1\}^{n/2}$ in the $n$-dimensional slice while nearly preserving the proximity of $f$ to the space of degree-$d$ polynomials on this cube (with high probability). While the embedding is simple and natural, the analysis involves a careful induction (seen in some prior works on low-degree testing) with a novel use of a basis of degree-$d$ polynomials on slices (from a work of Anstee, R\'{o}nyai and Sali (Graphs and Combinatorics 2002)). Such a basis of functions is non-trivial and has several nice combinatorial and algebraic closure properties. We show how these properties are useful by using them to analyze our low-degree tests.
\end{abstract}


\tableofcontents

\newpage

\section{Introduction}\label{sec:intro}

The low-degree testing problem, namely the task of testing if a multivariate function given as an oracle is (close to) a low-degree polynomial with few oracle queries, has been a fundamental problem in theoretical computer science for over three decades. (Some of the earlier major milestones in this line of work include \cite{BLR, RubSud,AroraS,ALMSS,RazSafra,AroraSudan,AKKLR,KaufRon}.) This problem has been studied in many different settings, settings that vary depending on the domain of the function, the range of the function, on the relationship between the degree and field/group size, and on the coarseness/fineness of the analysis. In this work we consider this problem in yet another new setting, where the domain is a Boolean ``slice'' and the range is an Abelian group. Here a ``slice'' is a subset of points in the Boolean cube $\{0,1\}^n$ all of the same Hamming weight --- we use $\{0,1\}^n_k$ to denote the $k$-th slice, namely points of Hamming weight $k$. A polynomial over a group $G$ has coefficients from $G$ while the variables take on integer (in our case one of 0/1) values with evaluation defined in the natural way. (For simplicity the reader may simply think of $G$ as the reals while reading this section.)

The slice is often an important domain to understand, especially in the context of average case analysis of algorithms. Just as the uniform distribution over $\{0,1\}^n$ models a uniformly random graph and the $p$-biased distribution models a graph drawn from $G_{n,p}$, the slice naturally corresponds to a random graph drawn from $G(n,m)$ --- a graph on $n$ vertices with exactly $m$ edges.

Our main motivation for studying this setting however is technical, to see what kind of tools can be brought to bear on this problem in our setting. Indeed low-degree testing has been a great source of connections to powerful mathematical tools, including studies of (higher order) Fourier analysis, affine invariance, high-dimensional expansion, (global) hypercontractivity. The two directions of prior works that come closest to our setting are from the works of David, Dinur, Goldenberg, Kindler and Shinkar~\cite{DDGKS17} and Kalai, Lifshitz, Minzer and Ziegler~\cite{KLMZ} who consider functions from the slice to $\F_2$; and the works of Bafna, Srinivasan and Sudan~\cite{BSS} and Amireddy, Srinivasan and Sudan~\cite{ASS} who consider the setting of functions from the Boolean cube to Abelian groups. It turns out neither of these settings capture our setting adequately and indeed even a natural test is not obvious. The latter works rely strongly on the ability to set variables independently allowing tests of $n$-dimensional functions to sample points from a tensor product space. This allows them to invoke hypercontractivity which is well analyzed in such settings. The former works~\cite{DDGKS17,KLMZ} roughly pick a random linear subspace of $\F_2^n$ and condition it on being in the middle slice $\{0,1\}^n_{n/2}$. While this conditioning is satisfied only with inverse polynomial probability in $n$ (for $O(1)$ dimensional spaces) that is still a relatively high probability event and allows them to build dense models that are closer to $\F_2^n$ and lift the behavior of tests from there. If the range however is $\R$ (or even $\F_3$), the analogous test would pick subspaces in $\R^n$ (or $\F_3^n$) and condition it on being in $\{0,1\}^n_{n/2}$, which of course happens with zero probability (or exponentially small in $n$). With such low probability occurrences it is hard to envision dense models capturing the behavior of the tests.

Nevertheless in this work we are able to design a low-degree test for our setting. Namely given a degree parameter $d$ we design a test $T_d$ that makes $O_d(1)$ queries to an oracle for a function $f:\{0,1\}^n_{n/2} \to G$ and accepts degree $\leq d$ polynomials while rejecting functions that are $\epsilon$-far with probability $\Omega_{d}(\varepsilon)$. The key ingredient in our test is to embed a copy of an $n/2$-dimensional cube in the $n/2$-slice in $n$ dimensions. While this embedding is quite natural, it is a priori unclear how to analyze it, and to do so, we have to blend an inductive argument going back to an analysis of low-degree tests due to Bhattacharyya, Kopparty, Schoenebeck, Sudan, and Zuckerman~\cite{BKSSZ} with a very special basis of polynomials of degree at most $d$ over the slice from a work of Anstee, R\'{o}nyai and Sali~\cite{ARS}. We elaborate more on the basis below, but note first that while the inductive approach is natural it was not a priori clear it would work and indeed the question of whether such an approach could work is raised explicitly in \cite{DDGKS17}.

To understand the need for a special basis, note that a typical analysis of properties of low-degree polynomials uses the standard monomial basis, and a common step in the reasoning is that if some coefficient is non-zero then the polynomial is non-zero. In our setting the standard set of monomials is clearly dependent and does not form a basis.
(To see this, note that for every \(e \geq 1\), the polynomials $\sum_{i=1}^n x_i^e - (n/2)$ vanish over the slice for every group.
More complex polynomials also turn out to be zero making the basis challenging to construct.) It turns out that a nice basis for polynomials is available in the literature, though it is not widely known in the CS literature. The reasons this basis is useful are that it has many natural and desirable properties such as the following.
\begin{itemize}
	\item It is a monomial basis (i.e. spanned by monomial functions) and is a basis for polynomials over any Abelian group. There are other bases e.g. the space of homogeneous polynomials of a given degree that form a basis over fields of characteristic $0$ but not over other fields.
	\item It is graded, i.e. basis elements of degree at most $d$ generate polynomials of degree at most $d$ for each $d$; and downward-closed, i.e. factors of monomials in the basis are also in the basis.
	\item It is closed under a weak form of restriction. This roughly translates to the fact that the basis for degree-$d$ polynomials over $\{0,1\}^{n}_{n/2}$ specializes to a basis for degree-$d$ polynomials over $\{0,1\}^{n-1}_{n/2}$ when $x_1$ is set to $0$. If $d=1$ and the basis were chosen to be $1,x_2,\ldots,x_{n}$ then this property does not hold, while choosing the basis to be $1,x_1,\ldots,x_{n-1}$ does satisfy this property.

\end{itemize}


In what follows we describe our main results before expanding on the proof techniques.


	






\subsection{Our results}

We show the following.

\begin{theorem}[{\bf Low-degree test over slice}]\label{thm:ldt-slices}
	For every Abelian group $G$, degree parameter $d\ge 0$, even integer $n \ge 2$, there exists a $\exp(d^{O(1)})$ query test $\cTslice$ such that for all functions $f:\bool^n_{n/2} \to G$, we have
	\begin{itemize}
		\item (Completeness) If $f$ is degree-$d$, then $\cTslice$ accepts $f$ with probability $1$.
		\item (Soundness) For every $\varepsilon>0$, if $f$ is $\varepsilon$-far from degree-$d$, then $\cTslice$ rejects $f$ with probability $\Omega_{d}(\varepsilon)$.
	\end{itemize}
\end{theorem}

\noindent We note that the rejection probability of the above test can be made $\Omega(\varepsilon)$ by repeating the test an appropriate $O_d(1)$ (in fact, $\exp(d^{O(1)})$) number of times. Furthermore, the test $\cTslice$ of~\Cref{thm:ldt-slices} can be easily described as follows (a more formal description is given in~\Cref{sec:overview}):
\begin{enumerate}
	\item Choose a perfect matching $M$ over the complete graph on $n$ vertices uniformly at random.
	\item Obtain oracle access to the restricted function $f^{(M)}:\bool^{n/2} \to G$ defined by setting $x_j = 1-x_i$ for all edges $(i,j)$ of the matching $M$ where $i<j\in [n]$.
	\item Test if $f^{(M)}$ is degree-$d$ using the low-degree tests of~\cites{BSS, ASS}.
\end{enumerate}

\begin{remark}
	We note that for $G=\Z_2$, a lower bound of $\Omega(2^{d})$ on the number of queries needed in~\Cref{thm:ldt-slices} follows by noting that the dual distance of the underlying code is at least $2^{d+1}$,\footnote{This follows by observing that the code underlying degree-$d$ evaluations over the slice is obtained by puncturing the code underlying degree-$d$ evaluations over the entire Boolean cube. In particular, the dual distance of the latter code is at least $2^{d+1}$ and dual distance cannot decrease on puncturing.} and then using the approach of~\cite{AKKLR} who showed a $\Omega(2^{d})$ lower bound for the Boolean cube setting.
\end{remark}

\begin{remark}
	\label{rem:new}
	For the third step, we use the low-degree test given by~\cite{BSS,ASS} to handle functions defined on the entire Boolean cube, although using the tester of~\cite{AKKLR,BKSSZ} suffices for the case of $\Z_2$. Even for the case of $d=1$ and $G=\Z_2$, our test and its analysis differ from the previous works~\cite{DDGKS17, KLMZ} on this problem. We also think that this new approach simplifies the conceptual landscape of this problem by giving a black-box reduction from the slice setting to the (much more well-studied) cube setting.
\end{remark}

Furthermore, we demonstrate the power of our techniques by showing a similar result for {\em imbalanced} slices (i.e., for $\bool^n_k$ where \change{$k\ne n/2$}). In fact, we use our low-degree test for the balanced slice to do this. We closely follow that approach of~\cite{DDGKS17}, which handles the $d=1$ case (over $\Z_2$).
This is done by ``reducing'' the problem to a {\em direct product testing} problem which was already solved by Dinur and Steurer~\cite{DinurSteurer}. In our case, we make use of a high dimensional variant of this given by Dinur, Filmus and Harsha~\cite{DFH}. In that sense, this is again a reduction to a problem that has been well-studied. However, performing this reduction for degree greater than $1$ again involves many nuances which we resolve by using the monomial basis described in the previous paragraphs.

\begin{theorem}[{\bf Imbalanced slices}]\label{thm:unbal-slices}
	For every Abelian group $G$, degree parameter $d\ge 0$, positive integers $k\ge 4d$ and $n\ge \change{2}k$, there exists a $\exp(d^{O(1)})$ query test $\cT^n_{k}$ such that for all functions $f:\bool^n_k \to G$, we have
	\begin{itemize}
		\item (Completeness) If $f$ is degree-$d$, then $\cT^n_{k}$ accepts $f$ with probability $1$.
		\item (Soundness) For every $\varepsilon > 0$, if $f$ is $\varepsilon$-far from degree-$d$, then $\cT^n_{k}$ rejects $f$ with probability $\Omega_{d}(\varepsilon)$.
	\end{itemize}
\end{theorem}



\subsection{Technical contributions}




A linearity test for the slice was given by~\cite{DDGKS17} (for functions from $\bool^n_k$ to $\F_2$), which was later extended by~\cite{KLMZ} to higher degrees (also to the more challenging ``1\% regime'').\footnote{More precisely, \cite{KLMZ} show that if a function passes their test even with slightly non-trivial probability, then it has non-trivial agreement with a low-degree polynomial. The more standard setting, which we consider here, is the ``99\% regime'' where we only prove results for functions that pass a given low-degree test with high probability.} However, these tests do not seem to work when the co-domain is changed to a larger field (e.g.~$\R$) as discussed in the introduction. We instead give a reduction from slices to Boolean grids and use the low-degree tests over grids given by~\cite{BSS,ASS} (which work for any Abelian group as the co-domain) whose analyses follow an iterative argument similar to the  tests of~\cite{BKSSZ,HSS}. 

Changing the structure of both the domain (from $\bool^n$ to slices) and the co-domain (from $\F_2$ to general Abelian groups) of the functions being tested brings in several new challenges that we will sketch here. Some of these challenges can be solved by bringing together ideas from prior works that handle at most one of these two changes. However, for the other challenges, we supply new ideas that could also be of broader interest in studying low-degree tests and Boolean slices, both of which are natural objects on their own and show up often in theoretical computer science. In particular, we believe that the monomial basis for slices of~\cite{ARS} that we use (see~\Cref{prop:monomial-basis}) can have other potential applications for problems concerning slices.

In order to point to our technical contributions and how the aforementioned basis helps us, we will need a quick overview of our proof (see~\Cref{sec:overview} for a more elaborate overview).
At an abstract level, we reduce the low-degree testing problem over slices to over the Boolean cube by choosing a certain random embedding of the $(n/2)$-dimensional cube in the $(n/2)$-th slice. Showing soundness of our test then boils down to showing that if a function is close to low-degree over many such random embeddings, then it is close to low-degree on the slice. Our high level idea is to treat the random embedding as an iterative process where we gradually reduce the ambient dimension of the embedding from $n$ to $n/2$ and argue that in each step, we don't lose too much in distance. With this setup, we now discuss three critical challenges posed in our setting and how the monomial basis of~\cite{ARS} lets us overcome these.

\begin{enumerate}
	\item {\bf Local characterizations.} To even show that our test rejects functions that are not low-degree with positive probability is non-trivial. While product domains admit simple (``local'') characterizations for being low-degree (à la Combinatorial Nullstellensatz), it is apriori not clear how to handle slice domains. However, once we have access to the above basis, it becomes fairly clear how to show that our tests have the required local characterizations for being low-degree (see~\Cref{prop:zero-error}).
	\item {\bf Iterative analysis.} Next, recall from the above high-level intuition of our test that we need to show that if a function is close to low-degree under many ``restrictions'' (e.g.~embeddings that we referred to above), then it is close to low-degree ``globally''. While traditional analyses of this step look at the monomial representation of the corresponding close-by polynomials to come up with the global polynomial by carefully ``gluing'' them, in our case, we need to look at all the polynomials in the above monomial basis to avoid inconsistencies. This requires us to ensure that the monomial basis behaves nicely under the restrictions that the test considers. And indeed we show that this holds true for our particular test (see~\Cref{clm:rijrji} and~\Cref{lem:restrn-basis}). This also raises interesting open questions about more general criteria for when such ``gluing'' of polynomials can always be done.
	\item {\bf Imbalanced slices.} Following~\cite{DDGKS17}, our approach to handle imbalanced slices (i.e., \change{say when $k < n/2$}) is to reduce the problem to the balanced slice by considering a certain random embedding of $\bool^{2k}_k$ in $\bool^n_k$. While~\cite{DDGKS17} implicitly uses the basis $1,x_1,\dots, x_{n-1}$ to handle the linear case, our extension of this argument to higher degree, among other things, involves using properties of the above monomial basis for slice (see~\Cref{clm:poly}).
\end{enumerate}

We note that there have been other natural (and closely related) bases for functions over slices considered in the past~(see e.g.~\cite{Filmus16}), which have more of an analytic flavor. We also remark here that slices (especially when $k\ll n$) are canonical {\em high-dimensional expanders} and indeed this was a major motivation behind the question of linearity testing considered in~\cite{DDGKS17, KaufmanLubotzky}. Thus, it seems natural to ask whether our higher degree version of the test (over general Abelian groups) reveals anything interesting about high-dimensional expanders.

\subsection{Open problems}

We list a few natural questions that arise from our work.

\paragraph{Gluing polynomials over subdomains.}

The key strategy used throughout this paper and in \cite{BKSSZ} can be
described in the following way:
Starting with some function
\(f: D \to G\)
defined on some domain \(D\),
we restrict to some randomly chosen subdomain \(D_{i} \subseteq D\),
and perform a low-degree test there,
either by using an induction hypothesis on the size of the domain,
or using that \(D_{i}\) is a simpler domain like a Boolean grid.

Such a test only works under the condition that if \(f\)
is far from being a degree-\(d\) polynomial,
then many of the restrictions
are also far from being a degree-\(d\) polynomial,
or conversely,
if \(f\) is close to polynomials \(P_{i}\) of degree-\(d\) on many domains \(D_{i}\),
then there exists a global polynomial \(P: D \to G\) that is close to \(f\).

This leads us to ask the following general question:
\begin{question}
	For an Abelian group \(G\), a set \(D \subseteq G^{n}\)
	and a collection of subsets \(D_{i} \subseteq D\),
	when does it hold
	that for any set of degree-\(d\) polynomials
	\(Q_{i} \in G^{\leq d}[x_{1}, \ldots, x_{n}]\)
	with \(Q_{i} \equiv Q_{j}\) on \(D_{i} \cap D_{j}\),
	there exists a polynomial \(Q \in G^{\leq d}[x_{1}, \ldots, x_{n}]\)
	satisfying \(Q \equiv Q_{i}\) on \(D_{i}\).
	That is, when can we ``glue'' degree-\(d\) polynomials defined on the subdomains \(D_{i}\),
	to a global one defined on the domain \(D\)?
\end{question}
As a key step in the proof,~\cite{BKSSZ} (raises and) answers the above question for the case when $D = \F_2^n$ with the subsets $D_i$ being hyperplanes over $\F_2$.
Similarly, in the local testing over grids works of~\cite{ASS,BSS}, the domain $D =\bool^n$ whereas $D_i$ are ``subcubes'' obtained by restricting some variables to constants. Similarly in a related work of Dinur, Filmus and Harsha~\cite{DFH}, the authors again use a certain {\em agreement tester}, to glue polynomials defined over restricted domains into one global polynomial.


In the same spirit, in this paper, we show that \(D=\bs\) with subsets being certain embeddings of the grid \(\bg[n/2]\) also has the desired structure  (see \Cref{prop:zero-error} and~\Cref{lem:global}).
Our key contribution has been finding a nice basis for the large
domain \(D\), such that the restrictions of the basis to the subdomains \(D_{i}\) are
well-behaved.
A better understanding of the general question could lead to further progress on low-degree tests over other domains.

\paragraph{Efficient decoding algorithms over the slice.}
Our results can be stated as showing that the code underlying low-degree evaluations over Boolean slices is locally testable. A recent work of Amireddy, Behera, Srinivasan and Sudan~\cite{ABSS25-SZ-Lemma} gives tight bounds on the distance of this code. Hence it is also natural to ask whether the same code also admits efficient decoding, local decoding, and list-decoding algorithms.

\subsection{Organization}
We state the required preliminaries in~\Cref{sec:prelims}, including the main monomial basis for functions on slices that we use. In~\Cref{sec:overview} we give an overview of our proof approach for testing over the balanced slice $\bool^n_{n/2}$ at a high level before starting out with the formal proofs. Then, in~\Cref{sec:zero-error}, we show that our low-degree test over the balanced slice has positive soundness. In order to strengthen this and show that the test has a constant soundness, we divide the analysis broadly into two cases: the ``small distance case'' (\Cref{sec:small-dist}) and the ``large distance case'' (\Cref{sec:large-dist}). Finally, we put everything from these sections together to finish the analysis of the low-degree test over the balanced slice in~\Cref{sec:put}. At the end in~\Cref{sec:unbal-slices}, we give our low-degree test and its analysis for imbalanced slices.

\section{Preliminaries}\label{sec:prelims}
We start with some basic notations and definitions.
Throughout the paper, we let $\bool^{n}_k$ denote the {\em $k$-th Boolean slice}, i.e., elements of $\bool^n$ of Hamming weight $k$, and let $G$ denote an arbitrary Abelian group. For $D\subseteq \bool^n$, we say that a function $f:D\to G$ is {\em degree-$d$} if it can be represented by a polynomial of degree at most $d$, where in this paper, a polynomial over an Abelian group $G$ refers to some linear combination of multilinear monomials of degree at most $d$ with coefficients coming from $G$. We let $\cP_d(D,G)$, or simply $\cP_d(D)$ when $G$ is clear from context, denote the family of degree-$d$ functions with domain $D$ and co-domain $G$. For $f,g:D\to G$, we denote their relative Hamming distance by $\delta(f,g):=\Pr_{\vecx \in D}[f(\vecx)\ne g(\vecx)]$.
For $f:D\to G$ and a family of functions with domain $D$ (with co-domain $G$), we also have the notation $\delta_{\cF}(f):=\min_{P\in \cF} \delta(f,P)$, which we refer to as the distance of $f$ from $\cF$. We say that $f$ is {\em $\varepsilon$-close} to $\cF$ if $\delta_\cF(f) \le \varepsilon$ and is {\em $\varepsilon$-far} otherwise. When $\cF$ is the family $\cP_d(D)$ and $D$ is clear from context, we denote
$\delta_d(f):=\delta_{\cP_d(D)}(f)$. For a domain $D$, a subset $S\subseteq D$ and function $f:D\to G$, we let $f|_S:S\to G$ denote the function $f$ restricted to the domain $S$.

For $S\subseteq [n]$ and $\vecx\in \bool^n$, we let $\vecx|_S$ denote the restriction of the string $\vecx$ to the coordinates of $S$ and for $\veca\in \bool^{|S|}$, we let $\veca^S$ denote the assignment to the coordinates indexed by $S$ by identifying $S$ with $[|S|]$ in increasing order of the coordinates. Similarly, we let $\bool^{S}_k$ denote the strings in $\bool^S$ of Hamming weight $k$. For disjoint subsets $S,T\subseteq [n]$, and strings $\vecx\in \bool^{S}$ and $\vecy\in \bool^T$, we let $\vecx \circ \vecy \in \bool^{S\cup T}$ denote their concatenation.

\subsection{Low-degree testing and slices}\label{subsec:ldt}

We will use the low-degree test of~\cite{BSS,ASS} as a crucial subroutine.

\begin{theorem}[{\bf Low-degree test over grids}, \cite{BSS,ASS}]\label{thm:grids}
	For every Abelian group $G$, degree parameter $d\ge 0$ and $n\ge 1$, there exists a $\exp(O(d))$ query test $\cTgrid$ such that for all functions $f:\bool^n \to G$, we have
	\begin{itemize}
		\item (Completeness) If $f$ is degree-$d$, then $\cTgrid$ accepts $f$ with probability $1$.
		\item (Soundness) For every $\varepsilon>0$, if $f$ is $\varepsilon$-far from degree-$d$, then $\cTgrid$ rejects $f$ with probability $\Omega(\varepsilon)$.
	\end{itemize}
\end{theorem}





We will start with a formal description of the final test we will use for the proof of our main theorem~\Cref{thm:ldt-slices}. For this, we first set up some notation.

For a function $f:\bool^n_{n/2} \to G$ over the slice and coordinates $i<j\in [n]$, we let $f^{(i,j)}$ denote the function over $(n-1)$ coordinates corresponding to setting the variable $x_j = 1-x_i$. We state the more formal and general definition of the restriction $f^{(i,j)}$ below. Before that, we note that we refer to the pair $(i,j)$ as an {\em edge} or a {\em restriction}, and we always follow the convention that the left coordinate of an edge is less than the right coordinate.

\begin{definition}[{\bf Edge restriction}]\label{defn:restriction}
	For a function $f:\bool^{n_1} \times \bool^{2n_0}_{n_0} \to G$ and $i<j\in [2n_0]$, let $f^{(i,j)}:\bool^{n_1+1}\times \bool^{[2n_0]\setminus \{i,j\}}_{n_0-1} \to G$ denote the restricted function (corresponding to the {\em edge restriction} $(i,j)$) defined for $\vecx\in \bool^{n_1+1}$ and $\vecy\in \bool^{[2n_0]\setminus \{i,j\}}_{n_0}$ as
	$$f^{(i,j)}(\vecx,\vecy) = f(\vecx^{[n_1]}, \vecy \circ x_{n_1+1}^{\{i\}}\circ (1-x_{n_1+1})^{\{j\}}).$$
\end{definition}

%

That is, we apply $f$ by taking the first part of the input to be the first $n_1$ indices of $\vecx$ and the second part of the input to be the string obtained by ``completing'' $\vecy$ by setting the $i$-th coordinate to $x_{n_1+1}$ and the $j$-th coordinate to be $1-x_{n_1+1}$. It is easy to verify that the final string lies in the domain of $f$ and hence $f^{(i,j)}$ is well-defined. We note that a single edge restriction reduces the number of coordinates by $1$.
Since substituting some variable $x_{i_1}$ of a polynomial to $x_{i_2}$ or $1-x_{i_2}$ cannot increase the degree of the polynomial, we have the following observation:

\begin{observation}\label{obs:rest}
	For $d,n_1\ge 0$ and $n_0\ge 1$, if
	$f:\bool^{n_1}\times \bool^{2n_0}_{n_0} \to G$
	is degree-$d$, then for all $i<j\in [2n_0]$, $f^{(i,j)}$ is degree-$d$.
\end{observation}

\begin{remark}
	Alternatively, one can define $f^{(i,j)}:S^{(i,j)} \to G$ as the function $f$ restricted to the subset $S^{(i,j)} := \{(\vecx,\vecy)\in \bool^{n_1}\times \bool^{2n_0}_{n_0}:y_i \ne y_j\}$. By the natural affine bijection between the two domains $\varphi:S^{(i,j)} \to \bool^{n_1+1} \times \bool^{2n_0\setminus \{i,j\}}_{n_0-1}$ defined by $\varphi(\vecx,\vecy) = (\vecx\circ y_i, \vecy|_{[2n_0]\setminus \{i,j\}})$, the degree of $f^{(i,j)}$ in both the definitions remains the same, and indeed we switch between the two definitions throughout the paper.
\end{remark}
Now, we define notation to handle successive edge restrictions by considering matchings.

\begin{definition}[{\bf Matching restriction}]\label{defn:match-restn}
	For $f:\bool^{n_1} \times \bool^{2n_0}_{n_0}\to G$ and a perfect matching $M$ over vertex set $[2n_0]$, we define the restricted function $f^{(M)}:\bool^{n_1} \times \bool^{\{i_1,i_2,\dots,i_{n_0}\}} \to G$ as follows:
	Let $(i_1,j_1),(i_2,j_2),\dots, (i_{n_0},j_{n_0})$ be the edges of $M$. Then, we define
	$$f^{(M)}(\vecx,\vecy):=f^{(i_1,j_1)(i_2,j_2)\dots(i_{n_0},j_{n_0})}(\vecx,\vecy).$$
\end{definition}

That is, we apply the restrictions $(i_1,j_1),\dots, (i_{n_0},j_{n_0})$ in this order.\footnote{Note, however, that the order does not matter. Applying the same restrictions in another order would also yield the same function $f^{(M)}$ up to a relabeling of coordinates.} In some contexts, for the sake of simplicity, we treat the second domain of $f^{(M)}$ as being $\bool^{n_0}$ instead of $\bool^{\{i_1,i_2,\dots,i_{n_0}\}}$. Even though our final tests only involve perfect matchings, the analysis goes through partial matchings. Therefore, we define for a partial matching $T$ over $[2n_{0}]$ with $n_2$ edges given by $(i_1,j_1),(i_2,j_2),$ $\dots,(i_{n_2},j_{n_2})$ in increasing order of their first coordinates, the restricted function $f^{(T)}:\bool^{[n_1]\cup \{i_1,\dots,i_{n_2}\}}$ $\times \bool^{[2n_0]\setminus \{i_1,j_1,\dots, i_{n_2},j_{n_2}\}}_{n_0-n_2}\to G$ similarly as follows:

$$f^{(T)}(\vecx,\vecy) := f^{(i_1,j_1)(i_2,j_2)\dots(i_{n_2},j_{n_2})}(\vecx,\vecy).$$

For the sake of simplicity in notation, we may sometimes relabel variables so that $f^{(T)}$ has domain $\bool^{n_1+n_2} \times \bool^{2n_0-2n_2}_{n_0-n_2}$.

\subsection{Bases for functions on slices}

As mentioned in~\Cref{sec:intro}, our low-degree tests use properties of certain monomial bases for functions defined on slices that we will introduce in this subsection.

We first note that functions over $\bool^n \to G$ have a basis given by the monomials $\cB(\vecx) = \{\prod_{i\in S} x_i:{S\subseteq [n]}\}$; importantly this monomial basis is {\em downward-closed} and {\em graded} (defined later in~\Cref{defn:props}). In order to work with polynomials over slices, we are going to need a similarly well-behaved monomial basis. While taking the same set of monomials $\cB(\vecx) = \{\prod_{i\in S} x_i : {S\subseteq [n]}\}$ is not even a basis (e.g.~the non-zero polynomial $\sum_{i=1}^n x_i -n/2$ evaluates to zero  identically on the balanced slice),~\Cref{prop:monomial-basis} establishes such a monomial basis along with the additional properties we need for our applications. In order to set it up, we start with the definition of a basis; we note that the formal definition below is important as we are dealing with arbitrary Abelian groups that may not be fields.

\begin{definition}[{\bf Basis and monomial basis}]\label{defn:basis}
	For a domain $D\subseteq \bool^n$ and $d\ge 0$, we say that a set of functions $\cB$ with domain $D$ and co-domain $\Z$ forms a {\em basis} for functions with domain $D$ and co-domain $G$ if it holds for every $f:D\to G$ that there exist unique elements $(\alpha_B)_{B\in \cB} \in G$ such that
	$$f\equiv \sum_{B\in \cB} \alpha_B \cdot B.$$

	We say that $\cB$ is a {\em monomial basis} (and denote it by $\cB = \cB(\vecx)$) if furthermore each function $B\in \cB$ can be expressed as a monomial of the variables $\vecx$.
\end{definition}

The following properties of monomial bases will be helpful for the analysis of our tests.

\begin{definition}[{\bf Downward-closed and graded bases}]\label{defn:props}
	We say that a monomial basis $\cB(\vecx)$ for a domain $D\subseteq \bool^n$
	is {\em downward-closed} if for all $T\subseteq  S\subseteq [n]$,
	we have that $\prod_{i\in S} x_i\in \cB(\vecx)$ implies
	$\prod_{i\in T} x_i \in \cB(\vecx)$.

	We say that a basis $\cB=\cB(\vecx)$ (that is not necessarily a monomial basis)
	for a domain $D\subseteq \bool^n$ is {\em graded} if it holds for all $d\ge 0$
	that $\cB^{\le d}:= \cB \cap \cP_d(D,\Z)$ forms a basis for $\cP_d(D)$.
	That is, the degree-$d$ functions in $\cB(\vecx)$
	form a basis for degree-$d$ functions over $D$, for each $d\geq 0.$
\end{definition}

\begin{fact}
	\label{fac:basis-bool-cube}
	For $D=\bool^n$, a downward-closed and graded monomial basis is given by $\cB(\vecx) = \{\prod_{i\in S} x_i:S\subseteq [n]\}$.
\end{fact}

Now we present an explicit monomial basis for the slice that we will make use of in the analysis of our low-degree tests. For this, we will need the definition of {\em ballot sequences} (or subsets with {\em ballot property}) which might remind the reader of the Catalan numbers.

\begin{definition}[{\bf Ballot property}]\label{defn:monomial-basis}
	We say that a subset $S\subseteq [n]$ (and the corresponding monomial $\prod_{i\in S} x_i$) has the {\em ballot property} if for all suffixes $T\subseteq [n]$ of $S$ (where we treat $S$ and $T$ as bit-strings corresponding to the indicator vectors of the subsets), it holds that the number of ones in $T$ is at most the number of zeros in $T$. More formally, $S$ has the ballot property if and only if
	for all $i\in [n]$, it holds that $$|S\cap \{i,\dots, n\}|\le |([n]\setminus S)\cap \{i,\dots, n\}|.$$
\end{definition}

The following lemma states that ballot sequences form a basis for functions on the slice.

\begin{proposition}[{\bf Monomial basis for slice},~{\cite[Corollary 4.4]{ARS}}]\label{prop:monomial-basis}
	For integers $n\ge 1$ and $0\le k\le n$ and Abelian group $G$,
	the monomials given by
	$$\cB_k^*(\vecx) = \bigg\{\prod_{i \in S} x_i: S \subseteq [n]\text{~is of size at most~}\min\{k,n-k\}\text{~and~} S\text{~has the {\em ballot property}}\bigg\},$$
	form a downward-closed and graded monomial basis for the set of functions from $\bool^n_k$ to $G$.
\end{proposition}

\paragraph{Example.} Consider the slice $\{0,1\}^5_2$. The above lemma asserts that the set of monomials
$$\cB^*_2(\vecx) = \{1,x_1,x_2,x_3,x_4,x_1x_2,x_1x_3,x_1x_4,x_2x_3,x_2x_4\}$$
forms a downward-closed, graded monomial basis for $\{0,1\}^5_2$. In particular,
$$\cB_2^{* \le 1}(\vecx) = \{1,x_1,x_2,x_3,x_4\}$$
forms a basis for degree-$1$ functions over $\bool^5_2$.

\begin{remark} A more intuitive description of the basis given by~\Cref{prop:monomial-basis} is as follows: Consider the bijection between the parenthesized expressions $\{(,)\}^n$ and subsets $S\subseteq [n]$ by mapping an index $i\in [n]$ to $``("$ if $i\in S$ and to $``)"$ otherwise. Then we note that subsets with ballot property correspond to parenthesized expressions that can be ``completed'' to a well-parenthesized expression by concatenating some string on its left. In particular, note that none of the monomials in the above basis contain the variable $x_n$ (this is not surprising since we can always replace $x_n$ with $k-\sum_{i=1}^{n-1}x_i$ over $\{0,1\}^n_k$).\end{remark}

We will use the following easy lemma (proved in~\Cref{prf:cart}) to construct graded monomial bases for mixed domains such as $\{0,1\}^{n_1}\times \{0,1\}^{2n_0}_{n_0}.$

\begin{lemma}[{\bf Basis of a Cartesian product}]\label{lem:cart}
	Let $D_1 \subseteq \bool^{n_1}$ and $D_2 \subseteq \bool^{n_2}$ and let
	$\cB_1=\cB_1(\vecx)$ and $\cB_2=\cB_2(\vecy)$ be bases for functions
	over $D_1$ and $D_2$ respectively. Then
	$$\cB(\vecx,\vecy) = \{b_1(\vecx)\cdot b_2(\vecy):b_1\in \cB_1 \text{~and~}b_2\in \cB_2\}$$
	is a basis for functions over $D_1\times D_2$. Moreover, if $\cB_1$ and
	$\cB_2$ are both monomial (resp.\ downward-closed, resp.\ graded) bases, then
	so is $\cB(\vecx,\vecy)$.
\end{lemma}

\subsection{Polynomial distance lemmas}
Here we recall the ``polynomial distance lemma'' (widely known as the ``Schwartz-Zippel lemma'') over the Boolean cube, and state similar results for slices (and their Cartesian product with the Boolean cube). As usual, we let $G$ denote an arbitrary Abelian group.

\begin{lemma}[{\bf Polynomial distance lemma over the Boolean cube},~{\cite{ore1922hohere, DL78, Zippel79, Schwartz80}}]\label{lem:SZ}
	Let $n\ge 1$, $d\ge 0$ be integers.
	If \(P: \bg[n] \to G\) is a non-zero degree-\(d\) function, then we have
	\begin{align*}
		\Pr_{\bx \in \bg[n]}[P(\bx) \neq 0]
		\;\geq\;
		\frac{1}{2^d}.
	\end{align*}
\end{lemma}

Now, we state a similar distance lemma for other domains that we will encounter in our low-degree tests, which we derive from a recent work of Amireddy, Behera, Srinivasan and Sudan~\cite{ABSS25-SZ-Lemma}.

\begin{lemma}[{\bf Polynomial distance lemma over slice},~{\cite[Theorem 1.1]{ABSS25-SZ-Lemma}}]\label{lem:dist}
	There exists an absolute constant $\alpha>0$ such that the following holds. Let $n_1\ge 0$ be an integer, $n\ge 1$ an even integer, $d\le n^\alpha$ be the degree parameter, and $D=\bool^{n_1}\times \bool^n_{n/2}$. If $P \in \cP_d(D,G)$ is a non-zero degree-$d$ function, then we have
	$$\Pr_{\vecx\in D}[P(\vecx)\ne 0] \ge \frac{1}{2^d}\cdot \paren{1-\frac{1}{n^\alpha}}.$$
\end{lemma}

Since~\cite{ABSS25-SZ-Lemma} do not consider Cartesian products with a Boolean cube, we provide a proof of~\Cref{lem:dist} in~\Cref{app:prelims}.

\begin{remark}\label{rem:sz-lem}
	We note that a weaker form of the above theorem was shown by~{\cite[Lemma 5.1.6]{ABPSS25}}, where the authors obtain a lower bound of $\binom{n-2d}{k-d}/\binom{n}{k}$ for the domain $D=\bool^n_k$. Consequently, for the domain $\bool^{n_1} \times \bool^{n}_{\lfloor n/2\rfloor}$, we get a bound of $1/2^{2d}$.
\end{remark}

\section{The Testing Algorithm and Proof Overview}\label{sec:overview}

With this setup, we give our main low-degree test over the balanced slice\footnote{For the imbalanced slices result, we refer to~\Cref{sec:unbal-slices}.} corresponding to~\Cref{thm:ldt-slices} below.
\vspace{3mm}

\begin{algorithm}[H]
	\caption{$\cTslice$ (Low-degree test over the balanced slice)}
	\label{algo:test-slice}

	\DontPrintSemicolon

	\KwIn{Oracle access to $f: \bs[n] \to G$, degree parameter $d$}

	Sample a uniformly random perfect matching $M$ over the complete graph on vertices $[n].$\;
	Let $M = \set{(i_{1}, j_{1}), \ldots, (i_{n/2}, j_{n/2})}$ \;
	Let $f^{(M)}: \bool^{n/2} \to G$ be the restriction of $f$ obtained by setting $x_{j_{k}} = 1 - x_{i_{k}}$ for every $k \in [n/2]$ \;
	Accept if and only if $\mathcal{T}_{\mathrm{grid}}(f^{(M)}, d)$ accepts.
\end{algorithm}

\vspace{3mm}


We note that the completeness of the test is immediate since every restriction of a degree-$d$ function remains degree-$d$ (by a repeated application of~\Cref{obs:rest}), so we can appeal to the completeness of $\cTgrid$ (see~\Cref{thm:grids}) to conclude that if $f:\bool^{n}_{n/2}\rightarrow G$ is degree-$d$, then $\cTslice$ accepts $f$ with probability $1$. We thus proceed to the main part of the theorem, which is the proof of the soundness of the above test.

\paragraph{Local and Robust Local characterizations.} We start by proving what is a basic `local characterization' property of the test $\cTslice$: namely, that if the test accepts with probability $1$, then the function $f$ is indeed a degree-$d$ function. While this may seem obvious, there are  natural situations~\cite{FriedlSudan, GuoKoppartySudan} where even some functions that are far from being degree-$d$ pass all the `obvious' local tests for degree-$d$ functions. In our setting, we are able to prove a local characterization quite simply by using the monomial basis for the slice from \Cref{prop:monomial-basis}. This is done in \Cref{prop:zero-error} below. More generally, we can also prove this over mixed domains $D = \{0,1\}^{n_1}\times \{0,1\}^{2n_0}_{n_0}$ for integers $n_0,n_1\geq 0.$

Using this, we can prove a `robust local characterization' over such mixed domains, which states that if a function $f: \{0,1\}^{n_1}\times\{0,1\}^{2n_0}_{n_0}\rightarrow G$ is close to a degree-$d$ function when restricted to any perfect matching on $[2n_0],$ then $f$ is close to a degree-$d$ function. This is done in \Cref{lem:final-few}.

\paragraph{Iterative analysis of~\cite{BKSSZ}.} With the above local characterizations in hand, we are ready to begin the analysis of $\cTslice$ in general. We use the general proof template of \cite{BKSSZ}, which is to view the process of restriction iteratively, and show that each step is unlikely to turn a function that is far from any degree-$d$ function to one that is `too close' to degree-$d.$ While~\cite{BKSSZ} used this to analyze low-degree tests over $\F_2^n,$ this idea has also been shown to be useful~\cite{BSS} for testing degree-$d$ functions from $\{0,1\}^n$ to any group $G$. Here, we use it to reduce the correctness of $\cTslice$ to low-degree testing over $\{0,1\}^n.$

More precisely, fix a function $f:\{0,1\}^{n}_{n/2}\rightarrow G$ that is an input to $\cTslice.$ We analyze what happens to the function $f$ when we pick the edges of the perfect matching $M$ one-by-one and restrict $f$ accordingly: recall that for each edge $\{i,j\}\in M,$ $f$ is restricted by setting the variable $x_j = 1-x_i$ (where $i < j$). This naturally leads to functions defined on mixed domains of the form $D = \{0,1\}^{n_1}\times \{0,1\}^{2n_0}_{n_0}.$ For the rest of the outline, we assume that we are dealing with such a function. Abusing notation, we will continue to call this function $f.$ The polynomial distance lemmas from \Cref{sec:prelims} imply that the space of degree-$d$ functions over such domains also yield codes of distance $\approx 2^{-d}$.

At the outset, we fix two distance thresholds $\varepsilon_0$ and $\varepsilon_1$ where $2^{-(d+1)} \gg \varepsilon_1 > \varepsilon_0 = \Omega_d(1).$ Our argument breaks into two parts: the `small distance case', where $\varepsilon := \delta_d(f) \leq \varepsilon_1$ \footnote{Recall that $\delta_d(f)$ is the distance from $f$ to the closest degree-$d$ function.} and the `large distance case' where $\varepsilon > \varepsilon_1.$

\subparagraph{Small distance case.} In the small distance case, $\varepsilon$ is small enough that there is a unique degree-$d$ function $P\in \cP_d(D)$ that is $\varepsilon$-close to $f.$  For a uniformly random perfect matching $M$ over $[2n_0]$, the function $f^{(M)}$ defines the restriction of $f$ to a random subset of its domain. The sampling properties of this random subset were recently investigated~\cite{ABSS25-SZ-Lemma} and using the second-moment bounds from that result, we can show that $\delta_d(f^{(M)}) \approx \varepsilon$ in expectation. Since the grid test $\cTgrid$ rejects with probability $\Omega(\delta_d(f^{(M)}))$, this finishes the proof for small $\varepsilon$. The formal proof is in \Cref{lem:close-distance}.


\subparagraph{Large distance case.} Finally, we come to the large distance case, i.e. $\varepsilon > \varepsilon_1$. Here, we argue that when we pick a random edge $(i,j)$ of the matching $M$ and restrict the function $f$ to $f^{(i,j)}$, the distance $\delta_d(f^{(i,j)})$ is not much smaller than $\varepsilon = \delta_d(f)$. In particular, we would like to argue that $\delta_d(f') > \varepsilon_0$ w.h.p.. Following~\cite{BKSSZ,BSS,ASS}, we prove this in the contrapositive by showing that if there are many  pairs $(i_1,j_1),\ldots, (i_t,j_t),$ such that $f^{(i_s,j_s)}$ is $\varepsilon_0$-close to degree-$d$, then $f$ is $\varepsilon_1$-close to degree-$d$, contradicting our assumption on $\varepsilon$ (see~\Cref{lem:restrn-large-dist}). To do this, we fix polynomials\footnote{These polynomials are uniquely defined as functions because $\varepsilon_0 \ll 2^{-(d+1)}$, which is within the unique decoding radius of the code.} $Q_s$ (for each $s\in [t]$) such that $f^{(i_s,j_s)}$ is $\varepsilon_0$-close to $Q_s,$ and show that there is a `global' degree-$d$ polynomial $Q$ that restricts to $Q_s$ under each of the appropriate restrictions (i.e. $Q^{(i_s,j_s)} = Q_s$ for all $s$). The existence of such a $Q$ is typically the most intricate part of the analysis, and once this is done, it is relatively straightforward to show that $f$ is $\varepsilon_1$-close to $Q$. Here, we are able to give a relatively simple argument by using the monomial basis from \Cref{prop:monomial-basis} to reduce this to the case of $\{0,1\}^n$, which has already been solved by~\cite{BSS,ASS} (see~\Cref{clm:rijrji} and~\Cref{clm:final-agreement}). At a high-level, we exploit the fact that the monomial basis from \Cref{prop:monomial-basis} behaves very similarly to the standard basis of the Boolean cube under restrictions of the first few variables. For example, for $i < j \leq n-2d,$ if we have a polynomial $R$ expressed in the basis of the slice $\{0,1\}^n_{n/2},$ then applying a substitution of the form $x_j := 1-x_i$ to $R$ followed by multilinearization (i.e. setting $x_i^2 = x_i$) yields another polynomial $R'$ in the basis corresponding to the mixed domain $\{0,1\}\times \{0,1\}^{n-2}_{n/2-1}$ (see~\Cref{lem:restrn-basis}).


\section{Local Characterization of Low-Degree Functions}\label{sec:zero-error}


In this section, we show that if a function $f:\{0,1\}^n_{n/2}\rightarrow G$ is sufficiently close to degree-$d$ when restricted to any perfect matching $M$, then $f$ is itself close to degree-$d.$

\subsection{Matchings characterize low degree}\label{subsec:matchings-char}

We start by proving the case when $f^{(M)}$ is a low-degree function for each $M.$  We prove a more general form of this statement over mixed domains, as this will be useful later on in the proof.

\begin{proposition}[{\bf Characterization of low degree}]\label{prop:zero-error}
	For every $d\ge 0$ and function $f:\{0,1\}^{n_1}\times \{0,1\}^{2n_0}_{n_0}\to G$, $f$ is degree-$d$ if and only if for all perfect matchings $M$ over $[2n_0]$, $f^{(M)}$ is degree-$d$.
\end{proposition}

\begin{proof}
	We first handle the forward implication of the above equivalence which is the easier direction. If $f$ can be represented by a polynomial $P$ of degree at most $d$ over its domain, then we observe that we can obtain a polynomial for $f^{(M)}$ by substituting the variables $x_j=1-x_i$ in $P$, for all edges $(i,j)$ of the matching $M$; clearly the degree of $P$ cannot increase upon these substitutions and the resulting polynomial agrees with $f^{(M)}$ over its domain $\{0,1\}^{n_1 + n_0}$.

	Now, we move to the harder direction of the equivalence. The proof is algebraic and makes use of the graded monomial basis for functions over the slice discussed in~\Cref{sec:prelims} and the corresponding construction of the basis of the Cartesian product $\{0,1\}^{n_1}\times \{0,1\}^{2n_0}_{n_0}$ obtained by taking the Cartesian product $\mathcal{B}$ of the natural bases of $\{0,1\}^{n_1}$ (\Cref{fac:basis-bool-cube}) and $\{0,1\}^{2n_0}_{n_0}$ (\Cref{prop:monomial-basis}). We use $x_1,\ldots, x_{n_1}$ to index the $n_1$ variables corresponding to the first part of the input from $\{0,1\}^{n_1}$ and $y_1,\ldots, y_{2n_0}$ to index the variables corresponding to the slice.

	Assume that $f$ is not a degree-$d$ function. This means that the unique polynomial $P$ representing $f$ in the basis $\mathcal{B}$ has degree equal to $D$ where $d < D \leq n_1+n_0.$\footnote{Note that every monomial in $\mathcal{B}$ has degree at most $n_1+n_0.$} Viewing $P$ as a polynomial in $\vecy$ with coefficients that are polynomials in $\vecx$, we can write
	\[
		P(\vecx,\vecy) = \sum_{|S|\leq D} Q_S(\vecx) \cdot \prod_{i\in S} y_i,
	\] where $Q_S$ is a multilinear polynomial of degree at most $D-|S|$ for each $|S|$.
	Fix a $T$ such that $\deg(Q_T) + |T| = D$ and $T$ is as large as possible in the graded lexicographic order.\footnote{In the graded lexicographic order, $T_1 < T_2$ if either $|T_1| < |T_2|$ or $|T_1|= |T_2|$ and the largest element in the symmetric difference $T_1\Delta T_2$ lies in $T_2$.} Assume that $t := |T|.$

	We recall that by the definition of the basis (\Cref{prop:monomial-basis}), for each $i\in [2n_0]$, we have that
	\begin{equation}
		\label{eq:basis-prop}
		|T\cap \{i,\ldots, 2n_0\}|\leq |([2n_0]\setminus T)\cap \{i,\ldots, 2n_0\}|.
	\end{equation}
	In particular, this implies that there is a matching between $T$ and $[2n_0]\setminus T$ so that each element $i\in T$ is matched to some $j> i.$ Such a matching can be constructed by sorting the elements of $T = \{i_1 < \cdots < i_t\}$ and greedily matching $i_t$ with the largest element of $[2n_0]\setminus T$, $i_{t-1}$ with the second largest element and so on. Property (\ref{eq:basis-prop}) ensures that each $i_p$ is matched to some $j_p > i_p.$
	We extend this matching to a perfect matching $M$ on the set $[2n_0].$

	Assume that $M  = \{(i_1,j_1),\ldots, (i_{n_0},j_{n_0})\}$ where as above we assume that $T = \{i_1,\ldots, i_t\}$ and $i_p < j_p$ for each $p\in [n_0].$

	Clearly, $P^{(M)}$ has degree at most $D$. We claim in fact that the polynomial $P^{(M)}$ has degree equal to $D$, which would finish the proof of the proposition. To see this, let us define $P^{(M)}$ to be a polynomial in the variables $\vecx$ and $\vecz := (z_1,\ldots,z_{n_0})$ where $y_{i_p}$ is replaced by $z_{p}$ in $f^{(M)}$ and $y_{j_p}$ is replaced by $(1-z_p)$ in $P^{(M)}$ for each $p\in [n_0].$ (I.e. the restriction is defined as in \Cref{sec:prelims} but we rename the surviving $\vecy$-variables to $\vecz$-variables for notational simplicity.) We write $P^{(M)}$ as
	\[
		P^{(M)}(\vecx,\vecz) = \sum_{|S|\leq D} R_S(\vecx) \cdot \prod_{i\in S}z_i.
	\]
	By the definition of $P^{(M)}$, the coefficient $R_{[t]}(\vecx)$ of the monomial $z_1\cdots z_t$ in $P^{(M)}$ is given by
	\[
		R_{[t]} = \sum_{\substack{T': T'\in \binom{[2n_0]}{\leq D},  \\ \forall p\in [t]:\ |T'\cap \{i_p,j_p\}| = 1\\
		T'\cap \{i_{p+1},\ldots,i_{n_0}\} = \emptyset }} (-1)^{|T'\cap \{j_1,\ldots,j_p\}|}Q_{T'}.
	\]
	However, note that for any $T'\neq T$ indexing a summand on the right hand side above, it holds that $T' > T$ in the graded lexicographic order (this is because $j_p > i_p$ for each $p\in [D]$). By our choice of $T$, we see that $\deg(Q_{T'}) < D-t$ for all $T'\neq T$ in the above sum.\footnote{We define the degree of the Zero polynomial to be $-\infty.$} As $\deg(Q_T) = D-t$, this implies that $\deg(R_{[t]}) = D-t.$ This shows that $P^{(M)}$ indeed has degree exactly $D>d$.
\end{proof}

In the rest of this section, we will show a more robust version of the above result. Indeed, this is an intermediate step in the analysis of our low-degree test over the slice.

\subsection{Matchings robustly characterize low-degree}\label{subsec:matchings-robust-char}

We prove the following lemma,
which is a strengthening of~\Cref{prop:zero-error},
in that if a function over the slice is far from low-degree,
then not only is there a matching restriction under which it is not low-degree,
but in fact it remains far from low-degree under such a restriction.
As hinted above,
indeed in the final analysis of our low-degree test in~\Cref{sec:put},
we will further strengthen~\Cref{lem:final-few}
and show that a similar statement holds not just for {\em one} but for {\em many} matchings.
In fact, we will crucially use~\Cref{lem:final-few}
in the proof of~\Cref{thm:ldt-slices}
for intermediate domains that are a Cartesian product of a Boolean cube and a slice,
and that is the reason we state~\Cref{lem:final-few} for such mixed domains.

\begin{proposition}[{\bf Robust local characterization of low-degree}]\label{lem:final-few}
	Let $d\ge 1$ be the degree parameter and $n_1\ge 0, n_0 \ge d+1$
	be integers and
	$f:\bool^{n_1} \times \bool^{2n_0}_{n_0} \to G$.
	Then there exists a perfect matching $M$ over $[2n_0]$ such that we have
	$$\delta_d(f^{(M)}) \ge \min\{\delta_d(f),1/2^{d+3}\}.$$
\end{proposition}

\begin{proof}
	At a high level,
	the proof strategy is similar to that of~\Cref{lem:restrn-large-dist}.
	However, now the restrictions involve perfect matchings instead of edges.
	We will show a ``pairwise consistency'' claim in this setting as well and use it to construct a (``global'') degree-$d$ polynomial that is close to $f$ using degree-$d$ polynomials for the restrictions corresponding to the matchings.

	Let
	$D=\bool^{n_1}\times \bool^{2n_0}_{n_0}$,
	and $\delta\le 1/2^{d+3}$ be such that
	$\delta_d(f^{(M)}) \le \delta$ for all perfect matchings $M$.
	Then, we will show that $\delta_d(f) \le \delta$.
	This would then finish the proof of~\Cref{lem:final-few}.
	Let $S_M\subseteq D$ denote the subset of the original domain corresponding to the perfect matching $M$:
	$$S_M = \{ (\vecx,\vecy):y_i \ne y_j\text{~for all~} (i,j)\in M\}.$$
	By assumption, we know that for every perfect matching $M$,
	there exists a degree-$d$ function $P_M:S_M \to G$ such that
	\begin{align}\label{eqn:close-all-matchings}\delta(f|_{S_M},P_M) \le \delta.\end{align}
	Our goal is to use these ``local'' polynomials $P_M$ to come up with a ``global'' polynomial that is close to $f$.
	As the first step, we first show a local consistency among the local polynomials.

	\begin{claim}[{\bf Pairwise consistencies}]\label{clm:pairwise-last-few}
		For every pair of perfect matchings $M_1$ and $M_2$ over $[2n_0]$ that differ in exactly two edges, i.e., $|M_1\setminus M_2| = |M_2\setminus M_1| = 2$ (we call such matchings to be {\em adjacent}), we have
		$$P_{M_1}|_{S_{M_1}\cap S_{M_2}} \equiv P_{M_2}|_{S_{M_1}\cap S_{M_2}}.$$
	\end{claim}

	\begin{proof}
		We first note that
		$|S_{M_1}| = |S_{M_2}| = 2^{n_1 + n_0}$
		and
		$|S_{M_1} \cap S_{M_2}| = 2^{n_1 + n_0-1}$
		(since we must have $y_{i_1} = 1-y_{i_2} = y_{i_3} = 1-y_{i_4}$,
		where $i_1,i_2,i_3,i_4\in [2n_0]$
		denote the vertices incident to the edges in which $M_1$ and $M_2$ differ).
		Therefore,
		\begin{align*}
			\delta(P_{M_1}|_{S_{M_1}\cap S_{M_2}},P_{M_2}|_{S_{M_1}\cap S_{M_2}})
			 & \le \delta(f|_{S_{M_1}\cap S_{M_2}},P_{M_1}|_{S_{M_1}\cap S_{M_2}}) + \delta(f|_{S_{M_1}\cap S_{M_2}},P_{M_2}|_{S_{M_1}\cap S_{M_2}}) \\
			 & \le \frac{|S_{M_1}|}{|S_{M_1}\cap S_{M_2}|} \cdot \delta(f|_{S_{M_1}},P_{M_1}|_{S_{M_1}})                                             \\
			 & ~~~~~~+ \frac{|S_{M_2}|}{|S_{M_1}\cap S_{M_2}|} \cdot \delta(f|_{S_{M_2}},P_{M_2}|_{S_{M_2}})                                               \\
			 & \le 2\cdot (\delta(f|_{S_{M_1}},P_{M_1})+\delta(f|_{S_{M_2}},P_{M_2}))                                                                \\
			 & \le 4\cdot \delta.
		\end{align*}
		Now, identifying
		$S_{M_1} \cap S_{M_2}$
		with $\bool^{n_1+n_0-1}$ in the natural way (which will preserve degree-$d$ property of functions) and treating
		$P_{M_1}|_{S_{M_1}\cap S_{M_2}}$ and $P_{M_2}|_{S_{M_1}\cap S_{M_2}}$
		as functions over $\bool^{n_1+n_0-1}$, we observe that distance between the two functions is at most $4\delta < 1/2^{d}$.
		Now, using the distance lemma over the Boolean cube (\Cref{lem:SZ}),
		we thus conclude that
		$$P_{M_1}|_{S_{M_1}\cap S_{M_2}} \equiv P_{M_2}|_{S_{M_1}\cap S_{M_2}},$$
		finishing the proof of~\Cref{clm:pairwise-last-few}.
	\end{proof}

	We now use the above pairwise consistency claim to show that actually there is a ``global'' function $F:D\to G$ that agrees with each of the functions $P_M$ under the restriction $M$, for all perfect matchings $M$ over $[2n_0]$.
	\begin{claim}[{\bf Global function from pairwise consistencies}]\label{clm:global-from-pairs}
		There exists a function $F:D\to G$ such that for all perfect matchings $M$ over $[2n_0]$, we have
		$$F|_{S_M} \equiv P_M.$$
	\end{claim}
	\begin{proof}
		To guarantee such a function $F$, it suffices to show that for every pair of matchings $M_1$ and $M_2$ over $[2n_0]$ (that are not necessarily adjacent), if $(\vecx,\vecy)\in S_{M_1}\cap S_{M_2}$, then $P_{M_1}(\vecx,\vecy)=P_{M_2}(\vecx,\vecy)$. If $M_1$ and $M_2$ are adjacent, we are done by using~\Cref{clm:pairwise-last-few}. Otherwise, the idea is to move from $M_1$ to $M_2$ through a sequence of adjacent matchings, all of which contain the point $(\vecx,\vecy)$.

		More formally, assume, without loss of generality, that $\vecy = 0^{n_0} 1^{n_0}$ and $M_1 = \{(1,n_0+1),\dots, (n_0,2n_0)\}$ (so that $(\vecx,\vecy) \in S_{M_1}$). Now, since $S_{M_2} \ni (\vecx,\vecy)$ (otherwise, we have nothing to prove), all the edges of $M_2$ must be from $[n_0]$ to $[2n_0]\setminus [n_0]$. For any such matching $M=\{(1,i_1),\dots, (n_0,i_{n_0})\}$ from $[n_0]$ to $[2n_0]\setminus [n_0]$ (which we will say is {\em good}), let $\sigma(M)$ denote the permutation over $[n_0]$ that maps $j\in [n_0]$ to $i_j - n_0 \in [n_0]$. It is easy to see that $\sigma(M_1)$ corresponds to the identity permutation. Now for arbitrary good matchings $M$ and $M'$, suppose $\sigma(M')$ can be obtained from $\sigma(M)$ by swapping some two elements (say $\sigma(M)_i = \sigma(M')_j$ and $\sigma(M)_j = \sigma(M')_i$ where $i<j\in [n_0]$). Then the edges $(i,n_0 + \sigma(M)_i)$ and $(j,n_0 + \sigma(M)_j)$ are present in $M$ and similarly the edges $(i, n_0 + \sigma(M)_j)$ and $(j, n_0 + \sigma(M)_i)$ are present in $M'$; all the other edges of $M$ and $M'$ are identical. In other words, $M$ and $M'$ are adjacent. We now use the fact that one can get any arbitrary permutation of $[n_0]$ starting with the identity permutation by repeatedly swapping appropriate elements for a finite number of steps. Thus, there is a finite sequence of good matchings from $M_1$ to $M_2$, say $M_1 = M^{(1)},\dots, M^{(t)} = M_2$ such that the matchings $M^{(i)}$ and $M^{(i+1)}$ are adjacent for all $i\in [t-1]$. Since each matching in the sequence is good, $(\vecx,\vecy)\in S_{M^{(i)}}$ for all $i\in [t]$, which implies that $P_{M^{(i)}}(\vecx,\vecy) = P_{M^{(i+1)}}(\vecx,\vecy)$ for all $i\in [t-1]$ using~\Cref{clm:pairwise-last-few}. Hence, we conclude that $P_{M_{1}}(\vecx,\vecy) = P_{M_2}(\vecx,\vecy),$ which finishes the proof of~\Cref{clm:global-from-pairs}.
	\end{proof}

	We now finish the proof of~\Cref{lem:final-few} using~\Cref{clm:global-from-pairs}. In fact, we will show that $F$ is the desired degree-$d$ function that is $\delta$-close to $f$. First, we argue why $F$ has degree at most $d$: we recall from~\Cref{prop:zero-error} that a function $g:\bool^{n_1}\times \bool^{2n_0}_{n_0}\to G$ is degree-$d$ if and only if for all perfect matchings $M$ over $[2n_0]$, the corresponding restriction $g^{(M)}$ is degree-$d$.


	Since $P_M \equiv F|_{S_M}$ is degree-$d$, we know that $F^{(M)}$ is degree-$d$ for all $M$. Hence, we conclude that $F$ must be degree-$d$. Now, since sampling a perfect matching $M$ over $[2n_0]$ uniformly at random, and then a random point in the corresponding $S_M$ produces a uniformly random point in $D=\bool^{n_1} \times \bool^{2n_0}_{n_0}$, we have
	$$\delta(f,F) = \E_{M}[\delta(f|_{S_M},F|_{S_M})] = \E_M[\delta(f|_{S_M}, P_M)] \le \delta,$$ where the last inequality is using~\eqref{eqn:close-all-matchings}. This finishes the proof of~\Cref{lem:final-few}.
\end{proof}

\section{Small Distance Case}\label{sec:small-dist}
From the polynomial distance lemma (\Cref{lem:dist}),
we have that for large enough \(n_{0}\),
that a function
\(f: \bg[n_{1}] \times \bss[n_{0}] \to G\)
with
\(\delta_{d}(f) \leq 2^{-(d+2)}\)
is closest to a unique polynomial \(P\) in \(\pds\).
This will simplify our analysis,
and we will in this section show that if
\(\delta_{d}(f) \leq 2^{-(d+2)}\),
then we can with high probability
say that the distance is preserved after restricting to a matching.
To prove this,
we need a small generalization of a sampling lemma from \cite{ABSS25-SZ-Lemma}, for which we provide a short justification in~\Cref{app:small-dist}.

\begin{lemma}\label{lemma:main-informal}
	There exists an absolute constant $\alpha > 0$ for which the following holds for all $\rho>0$.
	Let $g: \bs \to [0,1]$ be a function with
	\(\mathbb{E}[g(\bx)] = \rho\).
	Then,
	\begin{align*}
		\mathbb{E}_{M, \bx, \by}
		[g^{(M)}(\bx) \cdot g^{(M)}(\by)] \;
		\leq
		\; \rho^{2-2/n^{\alpha}} + \rho\cdot 2^{-n^{\alpha}},
	\end{align*}
	where \(M\) is a uniformly random perfect matching,
	and \(\bx, \by\) are uniformly random independent elements from $\bool^{n/2}$.

\end{lemma}

We are now ready to prove the main result of this section,
that distance to low-degree is preserved (in expectation)
after restricting to a random matching.

\begin{lemma}[{\bf Small distance lemma}]\label{lem:close-distance}
	There exist absolute constants $\alpha,C > 0$ such that for all integers $d\ge 1,n_1\ge 0$ and $n_0 \ge (d+1)^C$, and functions
	\(f: \bg[n_{1}] \times \bss[n_{0}] \to G\)
	with
	\(\delta_{d}(f) \leq   2^{-(4d+8)}\),
	we have
	\begin{align*}
		\E_{M}
		\left[\delta_{d}\left(f^{(M)}\right)\right]
		\geq \frac{\delta_{d}(f)}{2},
	\end{align*}
	where \(M\) is a uniformly random perfect matching over $[2n_0]$.
\end{lemma}
\begin{proof}
    Let $D:=\bool^{n_1} \times \bool_{n_0}^{2n_0}$ and $P \in \cP_d(D)$ be the degree-$d$ function that is $\delta$-close to $f$, where $\delta:=\delta_d(f)$. 
    We first note that 
    \begin{align}\label{eqn:exp-delta-one}\E_M\left[\delta(f^{(M)},P^{(M)})\right] = \delta,\end{align}
    since evaluating $f^{(M)}$ (resp.~$P^{(M)}$) at a random point in its domain (where $M$ is also random) corresponds to evaluating $f$ (resp.~$P$) at a random point in its domain. We will now show that 
    \begin{align}\label{eqn:exp-delta-two}
        \Pr_{M}\left[\delta(f^{(M)},P^{(M)})>\frac{1}{2^{d+1}}\right] \le \frac{\delta}{2}.
    \end{align}
    Since $\delta_d(f^{(M)}) = \delta(f^{(M)},P^{(M)})$ whenever $\delta(f^{(M)},P^{(M)})\le \frac{1}{2^{d+1}}$ (by \Cref{lem:SZ}), we infer that~$\eqref{eqn:exp-delta-one}$ and~\eqref{eqn:exp-delta-two} would together imply that 
    $$\E_{M}\left[\delta_d(f^{(M)})\right] \ge \delta-\frac{\delta}{2}=\frac{\delta}{2},$$
    which would then yield the bound of~\Cref{lem:close-distance}. The rest of the proof is to establish~\eqref{eqn:exp-delta-two}.

    Let $g:\bool^{2n_0}_{n_0} \to [0,1]$ defined by
    $$g(\vecx):=\delta(f(\cdot,\vecx),P(\cdot,\vecx)),$$ denote the distance between $f$ and $P$ for a given assignment to the last $2n_0$ coordinates. Thus, we have $\E[g(\vecx)] = \delta(f,P) =  \delta$ and $\E[g^{(M)}(\vecx)] = \delta(f^{(M)},P^{(M)})$ for all matchings $M$. Suppose towards a contradiction that~\eqref{eqn:exp-delta-two} does not hold. Then we have 
    $$\E_{M}\left[\delta(f^{(M)},P^{(M)})^2\right] \ge \frac{\delta}{2}\cdot \frac{1}{2^{2(d+1)}}= \frac{\delta}{2^{2d+3}},$$ and so
    \begin{align}\label{eqn:exp-lower-bd}
    \E_{M,\vecx,\vecy}\left[g^{(M)}(\vecx)\cdot g^{(M)}(\vecy)\right] = \E_M\left[ (\E[g^{(M)}(\vecx)])^2 \right] \ge \frac{\delta}{2^{2d+3}}. 
    \end{align}
    On the other hand,~\Cref{lemma:main-informal} implies (for suitable constants $\alpha,C>0$) that 
    \begin{align*}\label{eqn:exp-upper-bd}
        \E_{M,\vecx,\vecy}\left[g^{(M)}(\vecx)\cdot g^{(M)}(\vecy)\right] \le \delta^{2-2/n^\alpha} + \delta\cdot 2^{-n^\alpha} < \delta^{3/2} + \frac{\delta}{2^{n^\alpha}} \le \delta^{3/2} + \frac{\delta}{2^{2d+4}}.
    \end{align*}
    However, this contradicts the lower bound of~\eqref{eqn:exp-lower-bd} since $\delta \le 1/2^{4d+8}$. Therefore, we conclude that~\eqref{eqn:exp-delta-two} holds, finishing the proof of the small distance lemma (\Cref{lem:close-distance}).
\end{proof}

\section{Large Distance Case}\label{sec:large-dist}

We prove the following lemma, whose contrapositive roughly says
that if a function $f$ is far from degree-$d$, then many of the restricted functions $f^{(i,j)}$ are also far. The proof strategy closely resembles that of~\cite{BSS,ASS} who also handle similar ``edge'' restrictions corresponding to setting $x_j=1-x_i$, but becomes more intricate due to the nature of the domains of the functions being non-grids. Nevertheless, using properties of the graded monomial basis for functions on slices discussed in~\Cref{sec:prelims} (in particular, \Cref{prop:monomial-basis}), we are able to give a relatively simple reduction to the case of the Boolean grid already handled by~\cite{BSS,ASS}.

\subsection{Global function from local functions}\label{subsec:global-poly}

The main goal of this subsection is to show that if there are many restrictions of a function that are close to low-degree (via some ``local polynomials''), then the overall function is close to low-degree (via some ``global polynomial''). More formally, we show the following.

\begin{lemma}[{\bf Global from local}]\label{lem:restrn-large-dist}
	Let $d\ge 1$ be the degree parameter and $n_1\ge 0,n_0 \ge 4(d+2)$ be integers and $f:\bool^{n_1} \times \bool^{2n_0}_{n_0}\to G$. Suppose for some integer $t\in [d+2,n_0/4]$ and $0\le \delta\le 1/2^{2d+6}$, we have that
	\begin{align*}
		\Pr_{\substack{ i,j \in [2n_{0}] \\ i < j }}\brac{\delta_d(f^{(i,j)}) \; \leq \; \delta} \; \geq \; \dfrac{2t^{2}}{\binom{2n_{0}}{2}}.
	\end{align*}
	Then we have,
	\begin{align*}
		\delta_d(f) \; \leq \; \dfrac{1}{2^{t-1}} + \delta.
	\end{align*}
\end{lemma}

\begin{proof}

	We start with a high-level proof idea: The assumption in \Cref{lem:restrn-large-dist} tells us that there are several restrictions of $f$ that are close to degree-$d$ functions (on their respective domains). We refer to the close degree-$d$ polynomials as ``local'' polynomials (here local reflects the guarantee that closeness to degree-$d$ polynomial is guaranteed only for a subset of the domain). From these restrictions, we choose a subset of structured restrictions (we refer to them as ``star'' or ``matching''). The main technical step is to use the ``local'' polynomials and glue them together to construct a ``global'' polynomial of degree-$d$ (the global polynomial is going to be defined over the whole domain) that is close to $f$.

	Now we proceed with a formal proof.
	Call a restriction $(i,j) \in [2n_{0}]^2$ (where $i<j$) to be {\em good} if $\delta_d(f^{(i,j)}) \le \delta$.
	By assumption, we have at least $2t^2$ good restrictions.

	\paragraph{}Consider the undirected graph $\cG$
	over the vertex set $[2n_0]$ with edges given by pairs corresponding to the good restrictions. By assumption, we know $\cG$ has at least $2t^{2}$ edges. We claim that at least one of the following two cases occur\footnote{This is essentially the sunflower lemma applied to the set system corresponding to the edges of the graph $\cG$.}:
	\begin{itemize}
		\item \textbf{Case 1: (Star case)} There exists a vertex $v \in [2n_{0}]$ of degree at least $t$ (i.e, a {\em star} subgraph in $\cG$ with $t+1$ vertices), or
		\item \textbf{Case 2: (Matching case)} There exist at least $t$ pairwise disjoint edges (i.e., a {\em matching} of size $t$).
	\end{itemize}
	We argue this now. If $\cG$ has a vertex of degree $\geq t$, then we are done. Otherwise, we will construct a matching in $\cG$ that has at least $t$ edges. Assume that the degree of each vertex in $\cG$ is strictly less than $t$. We construct a matching of $\cG$ in the following greedy manner:
	Start with an empty matching $M$, add an arbitrary edge to $M$, and keep adding edges to $M$ until no more edges can be added.
	Let $M$ be the partial matching we get after the above algorithm stops. If $M$ has at most $2m$ vertices, then the total number of edges is at most $(2m)\cdot t$. Since the number of edges in $\cG$ is at least $2t^{2}$, the matching $M$ has at least $t$ edges. Thus we have shown that at least one of the above two cases happen.

	\paragraph*{}The rest of the proof is divided into two cases, depending on whether $\cG$ satisfies the Star case or the Matching case. Most of our arguments will be similar for the two cases, therefore we will handle the two cases in parallel and highlight any differences as they arise.
	For both cases, let the domain of $f$ be denoted by $D:=\{0,1\}^{n_1} \times \bool^{2n_0}_{n_0}$. If $\cG$ satisfies the Star case, we can assume without loss of generality that the star is formed by the edges $(1,2),(1,3), \ldots, (1,t+1)$. Similarly, if $\cG$ satisfies the Matching case, we can assume without loss of generality that the matching contains the edges $(1,2),(3,4),\dots, (2t-1,2t)$. We will now define a notation for the edges involved that we will use in both of the cases: For $i\in [t]$, let
	\begin{align*}
		e_{i} \; = \; \begin{cases}
			              (1,i+1)\text{, in the Star case,} \\
			              (2i-1,2i)\text{, in the Matching case.}
		              \end{cases}
	\end{align*}
	We shall abuse notation slightly and denote the restricted functions using $e_i$ instead of the pair corresponding to $e_i$. We then denote the domains of these restricted functions $f^{(e_i)}$ by
	\begin{align*}
		D_{i} \; := \begin{cases}
			            \; \bool^{n_1+1} \times \bool^{[2n_0]\setminus \{1,i+1\}}_{n_0-1},   & \text{ in the Star case,}     \\ \\
			            \; \bool^{n_1+1} \times \bool^{[2n_0]\setminus \{2i-1,2i\}}_{n_0-1}, & \text{ in the Matching case.}
		            \end{cases}
	\end{align*}
	By definition, for every $i \in [t]$, edge $e_{i}$ is a good restriction. For every $i \in [t]$, let $P_{i} \in \cP_d(D_i)$ denote a degree-$d$ function that is $\delta$-close to $f^{(e_i)}$.\\

	\noindent
	We will show that there exists a ``global'' degree-$d$ function $P\in \cP_d(D)$ that becomes identical to $P_i$ under the restriction $e_i$, for all $i\in [t]$. More formally, we will prove the following lemma.

	\begin{restatable}[{\bf Global function}]{lemma}{globallemma}\label{lem:global}
		For $f:D\to G$ (where $D=\bool^{n_1} \times \bool^{2n_0}_{n_0}$), and $(e_i=(a_i,b_i))_{i\in [t]}$ denoting the edges of the star or the matching, suppose $P_i$ is a degree-$d$ function that is $\delta$-close to $f^{(e_i)}$. Then there exists a degree-$d$ function $P\in \cP_d(D)$ such that we have
		\begin{align*}
			P^{(e_i)} \equiv P_i, \quad \text{ for all } \; i \in [t].
		\end{align*}
	\end{restatable}

	\noindent
	We prove \Cref{lem:global} in \Cref{subsec:gluing}. For now, we assume \Cref{lem:global} and proceed with the proof of \Cref{lem:restrn-large-dist}. We will now show that $f$ is close to $P$. More precisely, $\delta(f,P) \leq 1/2^{t-1} + \delta$, which will finish the proof of \Cref{lem:restrn-large-dist}.\\

	\noindent
	It will be convenient to identify the domains of the restriction functions (i.e., $D_i$) with subsets of the domain of $f$ (i.e., $D$).
	For $i\in [t]$, let $e_i=(a_i,b_i)$ (i.e., $(a_i,b_i)=(1,i+1)$ in the star case, and $(2i-1,2i)$ in the matching case) and $S_i \subseteq D$ denote the subset
	\begin{align*}
        S_i:= \{(\vecx,\vecy)\in D: y_{a_i} \ne y_{b_i} \}. 
    \end{align*}
	The final queries of our low-degree test lie in the $S_i$, so it is natural that they show up in the analysis. We now make some observations regarding restriction $(a_{i}, b_{i})$ and the corresponding domain $S_{i}$.

	\begin{observation}
		For every function $h: D \to G$ and for every $(\vecx,\vecy)\in S_i$, we have
		\begin{align*}
			h(\vecx,\vecy) \; = \; h^{(a_i,b_i)}(\vecx \circ y_{a_i}, \vecy^{[2n_0]\setminus\{a_i,b_i\}}),
		\end{align*}
		since $\vecy\in S_i$ implies that $\vecy|_{[2n_0]\setminus \{a_i,b_i\}} \in \bool^{[2n_0]\setminus \{a_i,b_i\}}_{n_0-1}$ and $y_{b_i} = 1-y_{a_i}$.
		Furthermore, for a uniformly random $(\vecx,\vecy)\in S_i$, the input to $h^{(a_{i}, b_{i})}$ (i.e,. $\left(\vecx \circ y_{a_i}, \vecy^{[2n_0]\setminus\{a_i,b_i\}}\right)$) is uniformly distributed over the domain of $h^{(e_i)}$ (i.e., $D_i$).
	\end{observation}

	\noindent
	Using the above observation, we have for all $i\in [t]$:
	\begin{align*}
		\delta(f|_{S_i},P|_{S_i}) \; = \; \Pr_{(\vecx,\vecy)\in S_i}[f(\vecx,\vecy)\ne P(\vecx,\vecy)] \; = \; \Pr_{(\vecx,\vecy)\in D_i}[f^{(e_i)}(\vecx,\vecy)\ne P^{(e_i)}(\vecx,\vecy)].
	\end{align*}
	Now since $P^{(e_i)}\equiv P_i$ (\Cref{lem:global}) and $\delta(f^{(e_i)},P_i)\le \delta$ by the definition of $P_i$, we get that
	\begin{align*}
		\delta(f|_{S_i},P|_{S_i}) \; = \; \Pr_{(\vecx,\vecy)\in D_i}[f^{(e_i)}(\vecx,\vecy)\ne P(\vecx,\vecy)] \; \leq \; \delta.
	\end{align*}

	\paragraph*{}Now, let $B_i \subseteq S_i$ (called {\em bad} points) denote the subset of $S_i$ where $f$ and $P$ differ. By the above bound, we have for all $i\in [t]$:
	\begin{equation}\label{eqn:bad-set}
		\dfrac{|B_i|}{|S_i|} \; \leq \; \delta.
	\end{equation}

	Moreover, we now show below that the union of $S_i$'s cover almost all of the domain $D$. In particular, we have the below upper bound for the intersection of their complements. In the star case, we have:
	\begin{align*}
		\bigg|\bigcap_{i=1}^t \overline{S_i}\bigg| \; = \; \bigg|\{(\vecx,\vecy)\in D:y_1 = y_2 = \cdots = y_{t+1}\}\bigg| \; \leq \; 2^{n_1}\cdot \bigg|\{\vecy\in \bool^{2n_0}_{n_0}:y_1 = y_2 = \cdots = y_{t+1}\}\bigg|.
	\end{align*}

	We will now need the following probability estimates.

	\begin{claim}[{\bf Probability estimates}]\label{clm:prob-estimates}
		For positive integers $t \le n/4$, we have the following probability estimates:
		\begin{itemize}
			\item $\Pr_{\mathbf{x} \in \bool^{2n}_{n}}[x_1 = x_2 = \dots = x_{t+1}] \; \leq \; {1}/{2^{t}}$,~and
			\item $\Pr_{\mathbf{x} \in \bool^{2n}_{n}}[x_{2i-1} = x_{2i} \; \text{ for all } \; i\in [t]] \; \leq \; {1}/{2^{t-1}}$.
		\end{itemize}
	\end{claim}

	\noindent Using \Cref{clm:prob-estimates}, we get,
	\begin{equation}\label{eqn:interscn-star}
		\bigg|\bigcap_{i=1}^t \overline{S_i}\bigg| \; \leq \; 2^{n_1} \cdot \binom{2n_0}{n_0} \cdot \frac{1}{2^t}.
	\end{equation}
	Similarly, in the matching case, we have the following upper bound:
	\begin{align*}
		\bigg|\bigcap_{i=1}^t \overline{S_i}\bigg| \; & = \; \bigg|\{(\vecx,\vecy)\in D:y_{2i-1} = y_{2i}\text{~for all }i\in [t]\}\bigg|                           \\
		                                              & \leq \; 2^{n_1}\cdot \bigg|\{\vecy\in \bool^{2n_0}_{n_0}:y_{2i-1} = y_{2i}\text{~for all }i\in [t]\}\bigg|.
	\end{align*}
	Again, using \Cref{clm:prob-estimates}, we get,
	\begin{equation}\label{eqn:interscn-matching}
		\bigg|\bigcap_{i=1}^t \overline{S_i}\bigg| \; \leq \; 2^{n_1} \cdot \binom{2n_0}{n_0} \cdot \frac{1}{2^{t-1}}.
	\end{equation}
	We now combine the above bounds to conclude that $f$ and $P$ are close. That is,
	\begin{align*}
		\bigg|\{(\vecx,\vecy)\in D : f(\vecx,\vecy)\ne P(\vecx,\vecy)\}\bigg| & \le \bigg|\bigg({\bigcap_{i=1}^t} \overline{S_i}\bigg) \cup \bigcup_{i=1}^t B_i\bigg| \tag{since $B_i$ is the subset of points in $S_i$ where $f$ and $P$ differ} \\
		                                                                      & \le |D|\cdot \frac{1}{2^{t-1}} + \bigg|\bigcup_{i=1}^t B_i\bigg|\tag{using \eqref{eqn:interscn-star} or \eqref{eqn:interscn-matching}}                            \\
		                                                                      & \le |D|\cdot \frac{1}{2^{t-1}} + \delta \cdot |D|. \tag{using \eqref{eqn:bad-set} and $B_i\subseteq S_i\subseteq D$}
	\end{align*}

	Hence, $\delta(f,P) \le 1/2^{t-1} + \delta$, thus finishing the proof of~\Cref{lem:restrn-large-dist}.
\end{proof}

\subsection{Gluing of local polynomials}\label{subsec:gluing}

In this subsection, we prove the existence of the global degree-$d$ function by gluing together the local degree-$d$ polynomials that each of the restrictions are close to,  i.e., we prove \Cref{lem:global}, which we recall below.

\globallemma*

\begin{proof}[Proof of~\Cref{lem:global}]
	We will work with the domains $S_i = \{(\vecx,\vecy)\in D:y_{a_i}\ne y_{b_i}\} \subseteq D$ instead of the domains $D_i$. In particular, we have for all $i\in [t]$ that $\delta(f|_{S_i}, Q_i) \le \delta$, where $Q_i\in \cP_d(S_i)$ is a degree-$d$ function defined as
	$$Q_i(\vecx, \vecy) = P_i(\vecx\circ y_{a_i},\vecy^{[2n_0]\setminus \{a_i,b_i\}}).$$
	Using these, the goal of this lemma is to construct a degree-$d$ function $P\in \cP_d(D)$ over the domain $D$ that agrees with $Q_i$ over the subset $S_i$, for all $i\in [t]$. This suffices because then for all $i\in[t]$ and $(\vecx,\vecy)\in D_i$, we will have
	\begin{align*}
		P^{(e_i)}(\vecx,\vecy) & = P(\vecx^{[n_1]}, \vecy\circ x_{n_1+1}^{\{a_i\}}\circ (1-x_{n_1+1})^{\{b_i\}})   \\
		                       & = Q_i(\vecx^{[n_1]}, \vecy\circ x_{n_1+1}^{\{a_i\}}\circ (1-x_{n_1+1})^{\{b_i\}}) \\
		                       & = P_i(\vecx,\vecy),
	\end{align*}
	where the first equality is using the fact that $(\vecx^{[n_1]}, \vecy\circ x_{n_1+1}^{\{a_i\}}\circ (1-x_{n_1+1})^{\{b_i\}})\in S_i$ and second equality is using the definition of $Q_i$. That would finish the proof of~\Cref{lem:global}.\\

	We have the following ``pairwise consistency'' claim for the restrictions $(e_i)_{i\in [t]}$:
	\begin{claim}[{\bf Pairwise consistencies}]\label{clm:overlap} For every $i\ne j\in [t]$, we have
		$$Q_i|_{S_{i}\cap S_j} \equiv Q_j|_{S_i\cap S_j}.$$
	\end{claim}

	We defer the proof of the above pairwise consistency claim to later. For now, we assume that the claim holds and show the following:

	\begin{claim}[{\bf Global polynomial from pairwise consistencies}]\label{clm:global-2}
		There exists a degree-$d$ function $P\in \cP_d(D)$
		such that for all restrictions $i\in [t]$, we have $$P|_{S_i} \equiv Q_i.$$
	\end{claim}

	\begin{proof}
		Note that if we are able to find such a ``global'' degree-$d$ function $P$, then it implies pairwise consistencies: that is for all $i\ne j\in [t]$,
		$$Q_i|_{S_i\cap S_j} \equiv P|_{S_i\cap S_j} \equiv Q_j|_{S_i\cap S_j}.$$ The high-level idea of~\Cref{clm:global-2} is to show the converse, which is much more involved. The idea is to express all the functions in a fixed monomial basis and construct (the coefficients of) $P$ by solving a linear system using the coefficients of the restricted functions; the pairwise consistency claim guarantees a solution to this linear system. In order to proceed with the formal proof, we need to set up some notation.\\

		\noindent
		Let $\cB_1 = \cB_1(\vecx) = \{\prod_{i\in S} x_i:S\subseteq [n_1]\}$ and $\cB_2 = \cB^*_{n_0}(\vecy)$ denote the downward-closed, graded monomial bases for $\bool^{n_1}$ and $\bool^{2n_0}_{n_0}$ given by \Cref{prop:monomial-basis} respectively. Using \Cref{lem:cart} we know that
		\begin{equation}\label{eqn:monbasis}
			\cB(\vecx,\vecy) = \{m_1(\vecx)\cdot m_2(\vecy):m_1\in \cB_1 \text{~and~}m_2\in \cB_2\}
		\end{equation} is a downward-closed, graded monomial basis for $D=\bool^{n_1} \times \bool^{2n_0}_{n_0}$. We use the shorthand notation $\cB = \cB_1 \times \cB_2$ to refer to the above equation.\\

		\noindent
		For $i\in [t]$,
		let $R_i(\vecx,\vecy)$ be a degree-$d$ polynomial over $G$ agreeing with $Q_i$ over its domain $S_{i}$ that only contains monomials from $\cB^{\le d}$
		(i.e., the monomials not in $\cB(\vecx,\vecy)^{\le d}$ have coefficient zero in $R_i$).
		For $i\ne j\in [t]$, let $R_{i,j}(\vecx,\vecy)$ be the degree-$d$ polynomial over $G$ obtained by substituting $y_{b_j} = 1-y_{a_j}$ and $y_{b_i} = 1-y_{a_i}$ in the polynomial $R_i$ (where we recall that $e_i=(a_i,b_i)$ are edges corresponding to a star or a matching) and simplifying/expanding out the resulting polynomial by multilinearizing the monomials since each variable only takes Boolean values; in particular, $R_{i,j}$ does not have the variables $y_{b_i}$ and $y_{b_j}$ appearing in it.
		Symbolically, we have
		\begin{align*}
			R_{i,j} \; = \; R_i|_{y_{b_j} = 1-y_{a_j}, \; y_{b_i} = 1-y_{a_i}}.
		\end{align*}
		For an arbitrary $i\in [t]$, from the definition of $R_{i,j}$, it is unclear whether the polynomials $R_{i,j}$ and $R_{j,i}$ are equal or not, because they are expressed in a particular basis $\mathcal{B}$. However, we now show that they are indeed equal.

		\begin{claim}[{\bf Pairwise consistency of local polynomials}]\label{clm:rijrji}
			For all $i\ne j \in [t]$,
			$$R_{i,j} = R_{j,i}.$$
		\end{claim}
		\begin{proof}
			We divide our argument into two cases --- the star case and the matching case.

			\paragraph{The star case:}Suppose the graph $\cG$ has a star with $(t+1)$ vertices, where the edges are $(1,2), \ldots, (1,t+1)$. Then for arbitrary $i\ne j\in [t]$ let us decompose the polynomial $R_{i,j}$ as follows:
			\begin{equation}\label{eqn:rij}
				R_{i,j} \; = \; y_1 \cdot R_{i,j}^{(1)} \, + \, R_{i,j}^{(0)},
			\end{equation}
			where $R_{i,j}^{(1)}$ and $R_{i,j}^{(0)}$ are polynomials in the variables $\vecx$ and $\vecy|_{[2n_0]\setminus {\{1,i+1,j+1\}}}$. Here $R_{i,j}^{(1)}$ is of degree at most $(d-1)$ and $R_{i,j}^{(0)}$ is of degree at most $d$. Pairwise consistency (\Cref{clm:overlap}) implies that whenever $y_{1} = 1-y_{i+1} = 1-y_{j+1}$ for some point $(\vecx,\vecy)\in D$, we must have
			\begin{equation}\label{eqn:qij}
				Q_i(\vecx,\vecy) = Q_j(\vecx,\vecy).
			\end{equation}
			Now observe that if we set $y_1 = 1-y_{i+1} = 1-y_{j+1} =0$ for $(\vecx,\vecy)\in D$, the domain of the remaining coordinates is $\bool^{n_1} \times \bool^{[2n_0]\setminus \{1,i+1,j+1\}}_{n_0-2}$. Hence, using \eqref{eqn:rij} and \eqref{eqn:qij}, we get for all $(\vecx,\vecy)\in \bool^{n_1} \times \bool^{[2n_0]\setminus \{1,i+1,j+1\}}_{n_0-2}$ that:
			\begin{align}\label{eqn:r0}R_{i,j}^{(0)}(\vecx,\vecy) = R_{j,i}^{(0)}(\vecx,\vecy).\end{align} Similarly for all $(\vecx,\vecy)\in \bool^{n_1} \times \bool^{[2n_0]\setminus \{1,i+1,j+1\}}_{n_0-1}$ (this corresponds to setting $y_1=1-y_{i+1}=1-y_{j+1}=1$) we obtain:
			\begin{align}\label{eqn:r1}
				R_{i,j}^{(1)}(\vecx,\vecy) + R_{i,j}^{(0)}(\vecx,\vecy) = R_{j,i}^{(1)}(\vecx,\vecy) + R_{j,i}^{(0)}(\vecx,\vecy).
			\end{align}

			\paragraph*{}We now use the properties of our monomial basis $\cB$ from~\eqref{eqn:monbasis} to conclude that the equations~\eqref{eqn:r0} and~\eqref{eqn:r1} imply that the corresponding polynomials on both sides of the two equations are identical. More specifically, let $\cB_2'=\cB^*_{n_0-2}(\vecy|_{[2n_0]\setminus \{1,i+1,j+1\}})$ be the monomial basis for $\bool^{[2n_0]\setminus \{1,i+1,j+1\}}_{n_0-2}$ given by~\Cref{prop:monomial-basis}.\\

			\noindent
			Since we are dealing with restrictions of slices obtained by fixing some coordinates, it will be useful to derive monomial bases of such restrictions using monomial bases of the original domain. In particular, the following lemma (which we prove later), roughly states that under appropriate restrictions, the basis monomials continue to be in the basis of the restricted domain.\\

			\begin{lemma}[{\bf Monomial bases of restrictions}]\label{lem:restrn-basis}
				Let $n\ge 1$ be an integer and $\cB=\cB^*_{\lfloor n/2\rfloor}(\vecx)$ denote the monomial basis for $\bool^n_{{\lfloor n/2 \rfloor}}$ given by~\Cref{prop:monomial-basis}. Let $1\le d\le n/2$ be an integer, $S\subseteq [n-2d]$ be a subset,  so that we have $d\le k\le |\overline{S}|-d$ for some integer $k$ (where $\overline{S}:=[n]\setminus S$).

				Let $\cB'=\cB^*_k(\vecx|_{\overline{S}})$ denote the monomial basis for $\bool^{\overline{S}}_{k}$ given by~\Cref{prop:monomial-basis} obtained by ordering the coordinates of $\overline{S}$ in increasing order. Then for all monomials $m(\vecx)\in \cB(\vecx)^{\le d}$ that do not contain any variable from $S$, it holds that
				\begin{align*}
					m(\vecx) \in \cB'(\vecx|_{\overline{S}})^{\le d}.
				\end{align*}
			\end{lemma}

			\noindent

			Now we apply~\Cref{lem:restrn-basis} (for the subset $S=\{1,i+1,j+1\}$) to claim that all the monomials contained in $R_{i,j}^{(0)}$ and $R_{j,i}^{(0)}$ are contained in $\cB_1\times \cB_2'^{\le d} \subseteq \cB_1 \times \cB'_2$. Now, since $\cB_1\times \cB_2'$ is a monomial basis for $\bool^{n_1} \times \bool^{[2n_0]\setminus \{1,i+1,j+1\}}_{n_0-2}$, using~\eqref{eqn:r0}, we conclude that indeed the polynomials $R_{i,j}^{(0)}$ and $R_{j,i}^{(0)}$ are identical to each other:
			$$R_{i,j}^{(0)} = R_{j,i}^{(0)}.$$

			We now note that all the monomials in $R_{i,j}^{(1)}$ and $R_{j,i}^{(1)}$, because of the downward-closed nature of the monomial basis $\cB$, are contained in $\cB^{\le d}$. Now, repeating the same argument as above, we conclude that:
			\begin{align*}
				R_{i,j}^{(1)} \, + \, R_{i,j}^{(0)} \; = \; R_{j,i}^{(1)} \, +  \, R_{j,i}^{(0)}.
			\end{align*}
			Putting everything together, we get
			\begin{align*}
				R_{i,j}  = y_1\cdot R_{i,j}^{(1)} + R_{i,j}^{(0)} = y_1\cdot R_{j,i}^{(1)} + R_{j,i}^{(0)} = R_{j,i}.
			\end{align*}
			This concludes the argument that in the star case, for every $i \neq j \in [t]$, the polynomials $R_{i,j}$ and $R_{j,i}$ are identical.

			\paragraph{The matching case:} Now, suppose we are in the matching case. We perform a similar argument as the star case by setting some variables to zeros and ones to conclude again that $R_{i,j} = R_{j,i}$ for all $i\ne j\in [t]$. Let us decompose the polynomial $R_{i,j}$ as follows:
			$$R_{i,j} = y_{2i-1} y_{2j-1} \cdot R^{(3)}_{i,j} + y_{2i-1} \cdot R_{i,j}^{(2)} + y_{2j-1} \cdot R_{i,j}^{(1)} + R_{i,j}^{(0)},$$ where $R_{i,j}^{(3)},R_{i,j}^{(2)},R_{i,j}^{(1)}$ and $R_{i,j}^{(0)}$ are polynomials in $\vecx$ and $\vecy|_{[2n_0]\setminus \{2i-1,2i,2j-1,2j\}}$ of degree at most $d-2$, at most $d-1$, at most $d-1$ and at most $d$ respectively. By the downward-closed property of $\cB$, all the monomials appearing in $R_{i,j}^{(3)},R_{i,j}^{(2)},R_{i,j}^{(1)}$ or $R_{i,j}^{(0)}$ are contained in $\cB^{\le d}$.

			Pairwise consistency of $Q_i$ and $Q_j$ implies that for every choice of $(\vecx,\vecy)\in D$ such that $y_{2i-1} = 1-y_{2i}$ and $y_{2j-1} = 1-y_{2j}$, we have
			\begin{align}\label{eqn:qijmatching} Q_i(\vecx,\vecy) = Q_j(\vecx,\vecy). \end{align}
			In particular, setting $y_{2i-1}=y_{2j-1}=1-y_{2i} =1-y_{2j}= 0$ and using the fact that $R_{i,j}$ and $R_{j,i}$  compute $Q_i$ and $Q_j$ under this restriction, we obtain for all $(\vecx,\vecy)\in \bool^{n_1} \times \bool^{[2n_0]\setminus \{2i-1,2i,2j-1,2j\}}_{n_0-2}$ that:
			\begin{align}\label{rij-equal}R_{i,j}^{(0)}(\vecx,\vecy) = R_{j,i}^{(0)}(\vecx,\vecy).\end{align}

			Letting
			$$\cB_2'=\cB{^*_{n_0-2}}(\vecy|_{[2n_0]\setminus \{2i-1,2i,2j-1,2j\}})$$
			denote the monomial basis for $\bool^{[2n_0]\setminus \{2i-1,2i,2j-1,2j\}}_{n_0-2}$ given by~\Cref{prop:monomial-basis}, we note that the monomials appearing in $R_{i,j}^{(0)}$ and $R_{j,i}^{(0)}$ are all contained in $\cB_1 \times \cB_2'^{\le d} \subseteq  \cB_1\times \cB_2'$ by applying~\Cref{lem:restrn-basis} (for the subset $S=\{2i-1,2i,2j-1,2j\}$). However, since $\cB_1\times \cB_2'$ is a monomial basis for $\bool^{n_1} \times \bool^{[2n_0] \setminus \{2i-1,2i,2j-1,2j\}}_{n_0-2}$, from~\eqref{eqn:qijmatching}, we get the following equality of polynomials: $$R_{i,j}^{(0)} = R_{j,i}^{(0)}.$$ By repeating the same argument for the other three assignments $(y_{2i-1},y_{2j-1})$ to $(0,1),(1,0)$ and $(1,1)$, we get the following additional relations respectively:
			\begin{align*}
				R_{i,j}^{(1)} + R_{i,j}^{(0)}                                 & =  R_{j,i}^{(1)} + R_{j,i}^{(0)}                                 \\
				R_{i,j}^{(2)} + R_{i,j}^{(0)}                                 & =  R_{j,i}^{(2)} + R_{j,i}^{(0)}                                 \\
				R_{i,j}^{(3)} + R_{i,j}^{(2)} + R_{i,j}^{(1)} + R_{i,j}^{(0)} & = R_{j,i}^{(3)} + R_{j,i}^{(2)} + R_{j,i}^{(1)} + R_{j,i}^{(0)}.
			\end{align*}
			Hence, overall we get for all $i\ne j\in [t]$ that $R_{i,j}^{(i')}=R_{j,i}^{(i')}$ for all $i'\in \{0,1,2,3\}$, and thus $$R_{i,j} = R_{j,i}.$$ This finishes the proof of~\Cref{clm:rijrji}.
		\end{proof}
		We now use the following claim about multilinear polynomials to derive the global degree-$d$ polynomial $P$ that agrees with every $R_i$ under the corresponding restriction $e_i$ (and so it agrees with $Q_i$ over $S_i$), for $i\in [t]$.
		\begin{claim}[{\bf Global polynomial from restrictions},~{\cite[Claim 3.8]{ASS}}]\label{clm:final-agreement}
			Let $n\ge 2t$ and $t\ge d+2\ge 3$ and $(e_i=(a_i,b_i))_{i\in [t]}$ (where $a_i < b_i \in [n]$) denote the edges of a star or a matching of size $t$ over vertices $[n]$. Suppose there exist degree-$d$ polynomials $(R_i(z_1,\dots, z_{n}))_{i\in [t]}$ over $G$ such that for all $i\ne j\in [t]$, we have that $$R_{i,j} = R_{j,i},$$ where $R_{i,j} = R_i|_{z_{b_i} = 1-z_{a_i}, z_{b_j} = 1-z_{a_j}}$ denotes the polynomial over $G$ obtained by substituting $z_{b_i} = 1-z_{a_i}$ and $z_{b_j} = 1-z_{a_j}$ in the polynomial $R_i(z_1,\dots, z_{n})$. Then, there exists a degree-$d$ polynomial $P(z_1,\dots, z_{n})$ such that for all $i\in [t]$, if we substitute $z_{b_i} = 1-z_{a_i}$ in both $P$ and $R_i$, we get the same polynomial, i.e., $$P|_{z_{b_i} = 1-z_{a_i}} = R_i|_{z_{b_i}=1-z_{a_i}}.$$
		\end{claim}
		The above claim is proved by~\cite{BSS,ASS} (for a more general case of functions defined over grids); we reproduce the proof in~\Cref{app:large-final-agreement} for the sake of completeness.
		Using the above claim for the polynomials $(R_i(\vecx,\vecy))_{i\in [t]}$ and $n=n_1+2n_0$, we conclude that there exists a ``global'' degree-$d$ polynomial $P$ over the variables $\vecx,\vecy$ such that $P(\vecx,\vecy) = R_i(\vecx,\vecy)$ for all $(\vecx,\vecy)\in D$ such that $y_{b_i} = 1-y_{a_i}$ (Note that the polynomial $P$ guaranteed by~\Cref{clm:final-agreement} can potentially have monomials outside the monomial basis for $D$, but that is not a concern for us). Since $R_i$ computes $Q_i$ over $S_i$, this gives us that $P|_{S_i} \equiv Q_i$ which finishes the proof of~\Cref{clm:global-2}.
	\end{proof}
	This also finishes the proof of~\Cref{lem:global} since we have found a $P\in \cP_d(D)$ agreeing with $Q_i$ over its domain for all $i\in [t]$.
\end{proof}

We now prove~\Cref{lem:restrn-basis} which roughly says that under appropriate restrictions, the basis monomials remain to be in the basis of the restricted domain.

\begin{proof}[Proof of~\Cref{lem:restrn-basis}]
	Let $T \subseteq \overline{S}$ denote the set of variables in $m(\vecx)\in \cB(\vecx)^{\le d}$. By~\Cref{prop:monomial-basis}, note that for all $i\in [n]$, we have
	\begin{align}\label{eqn:valid-t}|T\cap \{i,\dots, n\}| \le |([n]\setminus T)\cap \{i,\dots, n\}|.\end{align} We will show that $|T|\le \min\{k,|\overline{S}|-k\}$ and that $T$ has the ballot property (w.r.t.~$\overline{S}$). The first condition holds by our assumption on the range of $k$ (and the fact that $|T| \le d$). To show the second condition, let $i\in [n]$ be arbitrary and $U = \overline{S} \cap \{i,\dots, n\}$. We will show that
	\begin{align}\label{eqn:valid-u}|T\cap U| \le |(\overline{S} \setminus T)\cap U|\end{align} to finish the proof. We have two cases:
	\paragraph {\bf Case 1: $i > n-2d$.} Since $S\subseteq [n-2d]$, we have that $U = \{i,\dots, n\}$ and $\overline{S} \cap U = [n] \cap \{i,\dots, n\}$. Thus,~\eqref{eqn:valid-t} implies~\eqref{eqn:valid-u}.

	\paragraph {\bf Case 2: $i \le n-2d$.} We have $|T\cap U| \le |T| \le d$ since $T$ corresponds to a degree-$d$ monomial. On the other hand, $|(\overline{S} \setminus T)\cap U| \ge |\overline{S} \cap U| - |T\cap U| \ge 2d-d=d$, where we are using the fact that $\overline{S} \cap U \supseteq \{n-2d+1,\dots,n\}$ has size at least $2d$. This proves~\eqref{eqn:valid-u}.
\end{proof}

We now give a proof of the pairwise consistency claim below.

\begin{proof}[Proof of~\Cref{clm:overlap}]
	We will first handle the star case. Without loss of generality, let $i=1$ and $j=2$, i.e., the two restrictions correspond to the edges $(1,2)$ and $(1,3)$. Now, note that $S_i\cap S_j = E\cup F$, where we define
	$$E=\{(\vecx,\vecy)\in D:y_1=1,y_2=y_3=0\}\text{~and~}F=\{(\vecx,\vecy)\in D:y_1=0,y_2=y_3=1\}.$$
	We will show that $Q_1|_E \equiv Q_2|_E$ and $Q_1|_F \equiv Q_2|_F$ and thus conclude the proof in the star case. We only give the proof for $E$ as the proof for $F$ is analogous. Observe that $E$ and $F$ are disjoint sets of size equal to $2^{n_1}\cdot \binom{2n_0-3}{n_0-1}$. Hence we have the ratio $$\frac{|S_1\cap S_2|}{|S_1|} = \frac{2\cdot 2^{n_1}\cdot \binom{2n_0-3}{n_0-1}}{2\cdot 2^{n_1}\cdot \binom{2n_0-2}{n_0-1}} = \frac{1}{2}.$$ We now note that
	\begin{align*}
		\delta(Q_1|_E, Q_2|_E) \; & \le \; \dfrac{|S_1\cap S_2|}{|E|}\cdot \delta(Q_1|_{S_1\cap S_2},Q_2|_{S_1\cap S_2}) \quad \tag{using $E\subseteq S_1\cap S_2$}                \\
		                          & \le 2\cdot (\delta(f|_{S_1\cap S_2}, \; Q_1|_{S_1\cap S_2}) + \delta(f|_{S_1\cap S_2}, \; Q_2|_{S_1\cap S_2})) \quad \tag{triangle inequality} \\
		                          & \leq \; 2\cdot {\dfrac{\max\{|S_1|,|S_2|\} }{|S_1\cap S_2|}}\cdot (\delta(f|_{S_1},Q_1) \, + \, \delta(f|_{S_2},Q_2))                          \\
		                          & \leq \; 8 \cdot \delta \quad \tag{using $\delta(f|_{S_i},Q_i)\le \delta$}                                                                      \\
		                          & \le \frac{1}{2^{2d+2}}. \quad \tag{using $\delta\le 1/2^{2d+6}$}
	\end{align*}
	On the other hand, the polynomial distance lemma (in particular, \Cref{rem:sz-lem}) applied by treating $E$ as an appropriate Cartesian product of a Boolean cube and slice
	implies that if $Q_1|_E\not\equiv Q_2|_E$, $$\delta(Q_1|_E,Q_2|_E) \ge \frac{1}{2^{2d}} > \frac{1}{2^{2d+2}}.$$ Thus, we conclude that $Q_1|_E \equiv Q_2|_E$ (and similarly, $Q_1|_{F} \equiv Q_2|_{F}$) and this completes the proof of the pairwise consistency claim (\Cref{clm:overlap}) in the star case.

	We now turn to the matching case. Without loss of generality, again let $i=1$ and $j=2$, i.e., the two restrictions correspond to the edges $(1,2)$ and $(3,4)$. Now, we have $$S_1 \cap S_2 = E_1 \cup E_2 \cup E_3 \cup E_4,$$ where we define
	\begin{align*}
		E_1 = \{(\vecx,\vecy)\in D:y_1=1-y_2=0,y_3=1-y_4=0\}, \\
		E_2 = \{(\vecx,\vecy)\in D:y_1=1-y_2=0,y_3=1-y_4=1\}, \\
		E_3 = \{(\vecx,\vecy)\in D:y_1=1-y_2=1,y_3=1-y_4=0\}, \\
		E_4 = \{(\vecx,\vecy)\in D:y_1=1-y_2=1,y_3=1-y_4=1\}.
	\end{align*}
	Similar to the analysis in the star case, we have
	\begin{align*}
		\delta(Q_1|_{E_1},Q_2|_{E_1}) \; & \leq \; \frac{|S_1\cap S_2|}{|E_1|} \cdot \delta(Q_1|_{S_1\cap S_2},Q_2|_{S_1\cap S_2})                                   \\
		                                 & \leq \; 4 \cdot (\delta(f|_{S_1\cap S_2}, \; Q_1|_{S_1\cap S_2}) \, + \, \delta(f|_{S_1\cap S_2}, \; Q_2|_{S_1\cap S_2})) \\
		                                 & \leq \; 4\cdot \dfrac{\max\{|S_1|,|S_2|\}}{|S_1\cap S_2|}\cdot (\delta(f|_{S_1},Q_1)+\delta(f|_{S_2},Q_2))                \\
		                                 & \leq \; 16\cdot \delta \; \leq \; \dfrac{1}{2^{2d+2}}.
	\end{align*}
	However, the distance lemma (\Cref{lem:dist}) applied by treating $E_1$ as an appropriate Cartesian product of a Boolean cube and slice implies that if $Q_1|_{E_1}\not\equiv Q_2|_{E_1}$, then $$\delta(Q_1|_{E_1},Q_2|_{E_1}) \ge \frac{1}{2^{2d}} > \frac{1}{2^{2d+2}}.$$ Thus, we conclude that $Q_1|_{E_1} \equiv Q_2|_{E_1}$ (and by a similar argument, we obtain that $Q_1|_{E_j} \equiv Q_2|_{E_j}$ for all $j\in [4]$) and this completes the proof of the pairwise consistency claim (\Cref{clm:overlap}) in the matching case as well.
\end{proof}

This fully finishes the proof of the main lemma of this subsection, i.e.,~\Cref{lem:global} (and thus that of~\Cref{lem:restrn-large-dist} from~\Cref{subsec:global-poly}).

\section{Putting Everything Together: Balanced Slice}\label{sec:put}

Finally, we complete the analysis of our low-degree test over the slice (\Cref{thm:ldt-slices}).

\begin{proof}[Proof of~\Cref{thm:ldt-slices}]
	We recall our low-degree test $\cT_\text{slice}$ described in~\Cref{sec:overview} below:

	\begin{enumerate}
		\item Choose a perfect matching $M = \{(i_1,j_1),\dots, (i_{n/2},j_{n/2})\}$ over vertices $[n]$ uniformly at random.
		\item Obtain oracle access to the restricted function $f^{(M)}:\bool^{\{i_1,\dots, i_{n/2}\}} \to G$ defined by setting $x_{j_{i'}} = 1-x_{i_{i'}}$ for all $i'\in [n/2]$ (also defined in~\Cref{sec:prelims}).
		\item Check if $f^{(M)}$ is degree-$d$ using the test $\cT_\text{grid}$ given by~\Cref{thm:grids}.
	\end{enumerate}

	\paragraph*{Completeness:}If $f:\bool^{n}_{n/2}\rightarrow G$ is degree-$d$, then $f^{(M)}$ is degree-$d$ (by~\Cref{obs:rest}) for all $M$ and hence $\mathcal{T}_{\mathrm{slice}}$ always accepts.

	\paragraph*{Soundness:}For soundness of $\mathcal{T}_{\mathrm{slice}}$, it suffices to show that if $f$ is $\varepsilon$-far from degree-$d$ functions on the balanced slice, then $f^{(M)}$ is $\Omega_{d}(\varepsilon)$-far from degree-$d$ functions on the grid $\bool^{n/2}$ in expectation. Then, we can use the soundness guarantee of $\cT_\text{grid}$ to get the desired soundness guarantee of $\cTslice$. In particular, for suitable factors of $d$ hiding under the $\Omega_{d}(\cdot)$ notation, we have that
	\begin{align}\label{eqn:reject-prob}
		\Pr[\cTslice\text{~rejects~}f] & \geq \; \Pr_M\bigg[\cTgrid\text{~rejects~}f^{(M)} \bigg] \nonumber \tag{where the probability is both over $M$ and $\cTgrid$} \\
		                               & \ge \E_M\left[\Omega(\delta_d(f^{(M)})\right] \tag{using \Cref{thm:grids}}\\
                                       & \ge \Omega\left(\E_M[\delta_d(f^{(M)})]\right).
	\end{align}

	Thus, the core of the soundness analysis is to show that the expected value of $\delta_d(f^{(M)})$ is $\Omega_{d}(\varepsilon)$. Quantitatively, we will show that $$\E_M\left[\delta_d(f^{(M)})\right] \ge \varepsilon/\exp(d^{O(1)}).$$

	First, we will handle the extreme case corresponding to $n\le \poly(d)$: In this case, by applying~\Cref{lem:final-few}, we observe that there exists at least one matching $M$ such that $$\delta_d(f^{(M)}) \ge \min\{\delta_d(f),1/2^{d+3}\} \ge \varepsilon/2^{d+3}.$$ Since the total number of matchings is at most $n! \le \exp(n^{O(1)}) \le \exp(d^{O(1)})$, we get the desired lower bound of $\varepsilon/\exp(d^{O(1)})$ on the expected value of $\delta_d(f^{(M)})$.

	Hence, for the rest of the proof, we may assume that $d\ge 1$ and $n\ge (d+1)^C$ for a sufficiently large absolute constant $C$. We now have two cases based on whether $\delta_d(f)$ is ``small'' where we set the threshold to be $\varepsilon_1:=1/2^{4(d+2)}$.

	\paragraph{Small distance case:}
	Suppose $\delta_d(f) \le \varepsilon_1$.
	Applying the small distance lemma~(\Cref{lem:close-distance}), we get that
	$$\E_{M}\bigg[\delta_d(f^{(M)})\bigg] \ge \frac{\varepsilon}{2},$$
	as desired.

	\paragraph{Large distance case:}
	We have $\delta_d(f) > \varepsilon_1$. For $i'\in [n/2]$, let $f_{i'}$ be the function corresponding to the partial matching given by the first $i'$ edges of the test $\cTslice$, i.e.,
	$$f_{i'} = f^{(i_1,j_1)\dots(i_{i'},j_{i'})}.$$
	Following the same convention, let $f_0 = f$ so that $\delta_d(f_0) > \varepsilon_1$. Applying the (contrapositive of) large distance lemma (\Cref{lem:restrn-large-dist}) from~\Cref{subsec:global-poly} to the function $f_{i'}$ whose domain can be identified with $\bool^{i'}\times \bool^{n-2i'}_{n/2-i'}$, for all $i' \in [0,p]$ (where $p:=n/2-(d+1)^C$ for a sufficiently large absolute constant $C$) and partial matchings $(i_1,j_1),\dots, (i_{i'},j_{i'})$ over $[n]$ such that $ \delta_d(f_{i'}) > \varepsilon_1$, we have
	\begin{align}\label{eqn:large-cond}\Pr_{i_{i'+1}<j_{i'+1} \in [n]\setminus\{i_1,j_1,\dots, i_{i'},j_{i'}\} }\bigg[\delta_d(f_{i'}^{(i_{i'+1},j_{i'+1})}) > \varepsilon_0 \bigg] \ge 1-\frac{2t^2}{\binom{n-2i'}{2}},\end{align} where we choose $\varepsilon_0>0$ and $t\ge 1$ appropriately so that~\Cref{lem:restrn-large-dist} is applicable. In particular, we set $t:=20d$ and $\varepsilon_0 := 1/2^{20d}$ so that we have ${1}/{2^{t-1}} + \varepsilon_0 \le 1/2^{4d+8} = \varepsilon_1$.

	Now, for $i'\in [0,p]$, consider the event $\delta_d(f^{(M)}) \ge \varepsilon_0/2$ conditioned on the event $\varepsilon_0 \le \delta_d(f_{i'}) \le \varepsilon_1$. In this regime, since the small distance lemma (\Cref{lem:close-distance}) is applicable to $f_{i'}$ (since $n/2-i' \ge (d+1)^C$), for any arbitrary choice of the first $i'$ edges such that $\varepsilon_0 \le \delta_d(f_{i'}) \le \varepsilon_1$, we have
	\begin{align}\label{eqn:small-cond}\Pr\bigg[\delta_d(f_{i'}^{(i_{i'+1},j_{i'+1}),\dots, (i_{n/2},j_{n/2})}) \ge \frac{\varepsilon_0}{2}\bigg] \ge \frac{1}{2}.\end{align}

	Chaining the inequalities~\eqref{eqn:large-cond} and~\eqref{eqn:small-cond} from $i'=0$ to $i'=p$, we get
	\begin{align*}
		\Pr\bigg[\bigg(\delta_d(f^{(M)}) \ge \frac{\varepsilon_0}{2}\bigg) \cup \bigg(\delta_d(f_{p+1}) > \varepsilon_0\bigg) \bigg] & \ge \frac{1}{2}\cdot \prod_{i'=1}^p \bigg(1-\frac{400d^2}{(n/2-i'-1)^2}\bigg)     \\
		                                                                                                                             & \ge \frac{1}{2}\cdot \bigg(1-\sum_{i'=1}^p \frac{400d^2}{(n/2-i'-1)^2}\bigg)      \\
		                                                                                                                             & \ge \frac{1}{2}\cdot \bigg(1-\sum_{i'={n/2}-p-1}^\infty \frac{400d^2}{i'^2}\bigg) \\
		                                                                                                                             & \ge \frac{1}{2}\cdot \bigg(1-\frac{400d^2}{(d+1)^C-2}\bigg)                       \\
		                                                                                                                             & \ge \frac{1}{4}.
	\end{align*}
	Thus, either we have that \begin{align}\label{eqn:case1}\Pr[\delta_d(f^{(M)}) \ge \varepsilon_0/2] \ge 1/8,\end{align} or \begin{align}\label{eqn:case2}\Pr[\delta_d(f_{p+1}) > \varepsilon_0] \ge 1/8.\end{align} If we are in the first case (i.e.,~\eqref{eqn:case1}), we are done since we obtain $\E_M[\delta_d(f^{(M)})] \ge \varepsilon_0/16$. However, the second case (i.e.,~\eqref{eqn:case2}) needs some care. We observe that indeed if we set the final few edges of the matching at random, with good probability, there exists {\em some} choice of these few restrictions that has large distance from low-degree, and in particular, with good probability, the final restricted function is far from low-degree if the original function was far from low-degree. More precisely, we make use of~\Cref{lem:final-few} to conclude that
	\begin{align}\label{eqn:final-cond}
		\Pr\bigg[\delta_d(f^{(M)}) \ge \frac{\varepsilon_0}{2}\bigg] & \ge \Pr\bigg[\delta_d(f_{p+1}) > \varepsilon_0 \bigg] \cdot \Pr\bigg[\delta_d(f^{(M)}) > {\varepsilon_0}~\bigm\vert~ \delta_d(f_{p+1}) > {\varepsilon_0} \bigg] \\
		                                                             & \ge \frac{1}{8}\cdot \frac{1}{(n-2(p+1))!}                                                                                                                      \\
		                                                             & \ge \frac{1}{8(2(d+1)^C)!},
	\end{align} where we are using the fact that the number of perfect matchings over $n-2i'$ elements is at most $(n-2i')!$ and the fact that $\min\{\varepsilon_0,1/2^{d+3}\} =  \varepsilon_0$.

	Hence, in the large distance case, we have that
	$$\Pr_M\bigg[\delta_d(f^{(M)}) \ge \frac{\varepsilon}{2^{2d+7}} \bigg] \ge \Pr_M\bigg[\delta_d(f^{(M)}) \ge \frac{\varepsilon_0}{2}\bigg] \ge \frac{1}{8(2(d+1)^C)!}\ge \frac{1}{\exp(d^{O(1)})}.$$

	To wrap up the proof of~\Cref{thm:ldt-slices}, recall that we have shown in all cases, we have that
	$$\Pr_M\bigg[\delta_d(f^{(M)}) \ge \frac{\varepsilon}{2^{2d+7}}\bigg] \ge \frac{1}{\exp(d^{O(1)})},$$ and hence $$\E_M[\delta_d(f^{(M)})] \ge \frac{\varepsilon}{\exp(d^{O(1)})}.$$
	Thus, recalling the discussion at the beginning of the proof, we have that
	$$\Pr_M\bigg[\cTslice\text{~rejects~}f\bigg] \ge \frac{\varepsilon}{\exp(d^{O(1)})}.$$
	Hence, we get a low-degree test for the balanced slice making $\exp(d^{O(1)})$ queries and rejection probability $\Omega_d(\varepsilon)$, finishing the proof of~\Cref{thm:ldt-slices}.
\end{proof}

\section{Imbalanced Slices}\label{sec:unbal-slices}

In this section, we prove our low-degree testing theorem over slices that are not necessarily balanced, i.e.,~\Cref{thm:unbal-slices}.
We describe the low-degree test first.

\vspace{3mm}

\begin{algorithm}[H]
	\caption{$\cT^n_k$ (Low-degree test over $k$-th slice)}
	\label{algo:test-kslice}

	\DontPrintSemicolon

	\KwIn{Oracle access to $f: \bool^n_k \to G$ (where $n\ge 2k$), degree parameter $d\ge 0$}

	Sample a uniformly random subset $U$ of $[n]$ of size $2k$\;
	Let $f^{(U)}:\bool^U_{k} \to G$ be obtained from $f$ by setting the coordinates in $\overline{U}=[n]\setminus U$ to 0 \footnote{Minor abuse of notation since the restrictions corresponding to matchings $M$ are also denoted by $f^{(M)}$.}\footnote{Note that $f^{(U)}$ is a function over the balanced slice in $2k$ coordinates.}\;
	Accept $f$ if and only if $\cTslice(f^{(U)},d)$ accepts.
\end{algorithm}

\vspace{3mm}



Similar to the case of the balanced slice (\Cref{sec:large-dist}), the key lemma is to show that if \(\delta_{d}(f^{(U)})\)
is small with high probability (over the choice of $U$),
then \(\delta_{d}(f)\) is also small (i.e.,~\Cref{lem:restrict}). Our proof of this will again crucially use the monomial basis for slices we discussed in the earlier sections. Additionally, we will also use an agreement theorem of Dinur, Filmus and Harsha~\cite{DFH} (which can be thought of as a higher degree version of the direct product tester of~\cite{DinurSteurer} that was used by~\cite{DDGKS17} in their linearity test over slices) in ``gluing'' the local polynomials to obtain a ``global'' one. Now we state our lemma below.

\begin{lemma}\label{lem:restrict}
	Let \change{$n\ge 2k \ge 8d$ be positive integers}, \(f :\bool^{n}_{k} \to G\), and $\eta>0$. If it holds that
	\begin{align*}
		\E_{U \in \binom{[n]}{2k}} \left[ \delta_{d}\left(f^{(U)}\right)\right] \le \eta,
	\end{align*}
	then we have
	\begin{align}\label{eqn:upper-bd}
		\delta_{d}(f) \leq O_{d}(\eta).
	\end{align}
\end{lemma}

Assuming this, we immediately have the main theorem of this section:

\begin{proof}[Proof of~\Cref{thm:unbal-slices}]
	We claim that the test $\cT^n_k$ described above is the desired low-degree test.
	As usual, the completeness of the test is easy to prove: if $f$ can be represented by a degree-$d$ polynomial $P$,
	then a degree-$d$ representation for $f^{(U)}$ can be obtained by substituting the variables in $\overline{U}$ to 0 in $P$ (and then use the completeness of $\cTslice$).
	In order to show that the test is sound,
	assume that $f$ is $\varepsilon$-far from degree-$d$. Now, $\eta$ to be sufficiently small such that the term on the RHS of~\eqref{eqn:upper-bd} is at most $\varepsilon$. In particular, we can achieve this for some $\eta\ge \Omega_{d}(\varepsilon)$.
	Then by (the contrapositive of) \Cref{lem:restrict}, we conclude that
    \begin{align}\label{eqn:contrapos-imbal-key}\E_{U}\left[\delta_d\left(f^{(U)}\right)\right] \ge \eta.
    \end{align}
	Then
	we obtain, by invoking our low-degree test over balanced slice (\Cref{thm:ldt-slices}) for $f^{(U)}$ that:
	\begin{align}\label{eqn:imbal}
		\Pr[\cT^n_k\text{~rejects~}f] & \ge \Pr_U[\cTslice\text{~rejects~}f^{(U)}] \nonumber \\
        & \ge \E_U\left[\Omega_d(\delta_d(f^{(U)}))\right] \nonumber \\
		                              & \ge \Omega_{d}(\eta)
		= \Omega_{d}(\varepsilon), \tag{using~\eqref{eqn:contrapos-imbal-key}}
	\end{align}
	where the probability in the first inequality is both over the internal randomness of the test and $U$, thus finishing the proof.
\end{proof}

We now prove the key lemma used above.

\begin{proof}[Proof of~\Cref{lem:restrict}]
	At a high level, the proof is similar to the ``Global from local'' lemma (\Cref{lem:restrn-large-dist}) used for the balanced slice case.
	We use the shorthand notation $\delta_U := \delta_d(f^{(U)})$.
	Then for every subset \(U\), we pick a degree-$d$ function \(P_{U} \in \cP_d(\{0,1\}^U_k)\)
	closest to \(f^{(U)}\) (i.e., $\delta_U = \delta(f^{(U)},P_U)$).
	We then find a degree-$d$ function \(P: \bool^{n}_{k} \to G\),
	such that on many subsets $U$, we have \(P^{(U)} \equiv P_{U}\).
	We then lastly note that many of the points in \(\bool^{n}_{k}\)
	lie inside one of these subsets, which will be used to complete the proof. We provide more details next.

	\paragraph{Constructing global polynomial.}

	We will use the following theorem of~\cite{DFH} (instantiated for our setting) to construct a global polynomial. Recall that ${[m]\choose \le d}$ denotes the set of subsets of $[m]$ of size at most $d$, and for a function $g:{[m]\choose \le d}\to \Sigma$ and $S\in {[m]\choose \ell}$, we let $g|_S:{S\choose \le d}\to \Sigma$ denote the function $g$ restricted to subsets containing only the elements of $S$. 

	\begin{theorem}[{\bf Agreement theorem in high dimensions}, {\cite[\change{Theorem 5.1}]{DFH}}]\label{thm:agr-test}
		For $\gamma>0$ and positive integers $m,\ell,d$, $t=\lfloor3\ell/4\rfloor$ such that $d\le t<\ell \le m$, and any alphabet $\Sigma$, the following holds: Let $\paren{L_S:{S\choose \le d} \to \Sigma}_{S\in {[m]\choose \ell}}$ be an ensemble of local functions satisfying
		\begin{align}
			\Pr_{\substack{\change{T\in {[m]\choose t}} \\\change{T\subseteq S_1,S_2\in {[m]\choose \ell}}}} \bigg[L_{S_1}|_{\change{T}} \not\equiv  L_{S_2}|_{\change{T}}\bigg] \le \gamma.
		\end{align}
		Then the majority-decoded function $M:{[m]\choose \le d} \to \Sigma$ satisfies 
		\begin{align}
			\Pr_{S\in {[m]\choose \ell}}\bigg[L_S \not\equiv M|_S\bigg] \le \exp(d^{O(1)})\cdot \gamma,
		\end{align} where the majority-decoded function $M$ is the one given by the ``popular vote'', namely for each $A \in {[m]\choose \le d}$, we set $M(A)$ to be the most frequently occurring value among $(L_S(A))_{S\supseteq A}$ (breaking ties arbitrarily).
	\end{theorem}

	We will use the theorem in the following way.
	We first pick a subset \(D \subseteq [n]\) of size \(2d\),
	such that 
    \begin{align}\label{eqn:delta-u-eta}\E_{U \in {[n]\choose 2k}\text{~and~}U\supseteq D} \left[ \delta_U \right] \le \eta.\end{align} Such a \(D\) must exist by an averaging argument.
	Roughly speaking, the intuition for this conditioning on all good subsets containing a fixed subset $D$ is the following: Since we will be dealing with polynomials evaluated over slices and monomial bases for the same, it will be useful to have a slack of few (in this case, $2d$) coordinates in order to ensure that the polynomials have a unique representation. 

	Furthermore, we assume without loss of generality that
	\begin{align*}
		D =  \left\{n-2d+1, \ldots, n\right\}.
	\end{align*}

	If we now let \(m := n-2d\) and \(\ell:= 2k-2d\),\footnote{We will apply~\Cref{thm:agr-test} with these parameters.}
	then every subset \(U\) containing \(D\) is of the form
	\(S\cup D\), where \(S\) is a subset of \([m]\) of size \(\ell\).
	We have that every \(P_{S \cup D}\) is a polynomial defined over
	\(\left\{0,1\right\}^{S \cup D}_{k}\),
	which by \Cref{prop:monomial-basis} has a certain graded basis.
	We denote this basis by \(\mathcal{B}_{U}\).
	We can then write \(P_{U}\) as
	\begin{align}\label{eqn:pu}
		P_{U} = \sum_{b \in \mathcal{B}^{\leq d}_{U}} c_{b} \cdot  b
		= \sum_{C \in \binom{S}{\leq d}} p_{U,C}(\bx|_{D}) \cdot \prod_{i\in C} x_i,
	\end{align}
	where \(p_{U,C}(\bx|_{D})\)
	is the sum over the basis elements
	which share the same variables of \(S\),
	so \(p_{U,C} \in G^{\leq d}[\bx|_{D}]\) is a polynomial in the ``\(D\)-variables''
	of degree at most \(d\); in fact its degree will be at most $d-| C|$ as the degree of $P_U$ is at most $d$.
	We can then for every \(S \in \binom{[m]}{\ell}\)
	define a local function
	\begin{align}\label{eqn:ls-defn}
		L_{S}: \binom{S}{\leq d} & \to G^{\leq d}[\bx|_{D}] \nonumber  \\
		C   \quad                & \mapsto \; p_{S\cup D,C}(\bx|_{D}).
	\end{align} 
	We then want to apply \Cref{thm:agr-test} to this ensemble of local functions,
	so we have to get an upper bound on the probability that
	\(L_{S_{1}}|_{S_{1} \cap S_{2}} \not\equiv L_{S_{2}}|_{S_{1} \cap S_{2}}\).
	The following claim restates it	to something we can more easily work with.
	\begin{claim}\label{clm:poly}
		Let $S_1,S_2\in {[m]\choose \ell}$ and \(U_{1} = S_{1} \cup D\), \(U_{2} = S_{2} \cup D\).
		We have \(L_{S_{1}}|_{S_{1} \cap S_{2}} \equiv L_{S_{2}}|_{S_{1} \cap S_{2}}\)
		if and only if \(P_{U_{1}} \equiv P_{U_{2}}\).\footnote{Note that technically the domains of the functions $P_{U_1}$ and $P_{U_2}$ are different; however, by this notation, we mean that the two functions agree in the intersection of the two domains: $\bool^{U_1\cap U_2}_k$.}
	\end{claim}

	We are then left to upper bound the following quantity where $t:=\lfloor 3\ell /4\rfloor$:
	\begin{align}\label{eqn:local-agr}
		\gamma:=\Pr_{\substack{\change{T\in {[m]\choose t}} \\ \change{T\subseteq S_1,S_2\in {[m]\choose \ell}}}}\bigg[P_{S_{1}\cup D} \not\equiv P_{S_{2}\cup D}\bigg].
	\end{align}
	\change{Now for any $U_1,U_2\in {[n]\choose 2k}$ such that $u:=|U_1\cap U_2|\ge t+2d$ and $P_{U_1}\not\equiv P_{U_2}$}, by applying the polynomial distance lemma
	(in particular, the weak version given by \Cref{rem:sz-lem}) over the $k$-th slice of $U_1\cap U_2$,
	we get the following:
	\begin{align}\label{eqn:sz-u1u2}
		\Pr_{\vecx \in \bool^{U_1\cap U_2}_k}\bigg[P_{U_1}(\vecx) \ne P_{U_2}(\vecx)\bigg]
		\ge
		\frac{{\change{u-2d}\choose k-d}}{{\change{u}\choose k}}
		\ge
		\paren{\frac{\min\{k-d,\change{u-d-k}\}}{\change{u}}}^{2d}
		\ge
		\Omega_{d}(1),
	\end{align} where the last inequality is using $k\ge 4d$ and $\change{2k\ge u \ge t+2d \ge 3k/2}$.
	Combining~\eqref{eqn:local-agr} and~\eqref{eqn:sz-u1u2}, we then have
	\begin{equation}\label{eqn:agreement-1}
		\Pr_{\substack{\change{T\in {[m]\choose t}}\\ \change{T\subseteq S_1,S_2\in {[m]\choose \ell}}\\ \change{\vecx\in \bool^{(S_1\cap S_2)\cup D}_k}}}\bigg[P_{S_1\cup D}(\vecx) \ne P_{S_2\cup D}(\vecx)\bigg]
		\ge
		\gamma \cdot \Omega_d(1) = \Omega_d(\gamma) .
	\end{equation}
	We will now establish an upper bound on the LHS of the above expression:
	For $S\in {[m]\choose \ell}$ and \change{$k-2d\le i \le \ell$}, let $\cD_{S,\change{i}}$ denote the distribution over ${S\cup D \choose k}$ corresponding to picking a subset $\change{S'}\in {S\choose \change{i}}$ uniformly at random and then outputting $\vecx \in {\change{S'}\cup D \choose k}$ uniformly and independently.\footnote{\change{We can interchangeably also treat $\vecx$ as an element of $\bool^{S'\cup D}_k$.}}
	While $\cD_{S,\change{i}}$ need not be the uniform distribution over ${U \choose k}$ (where $U= S\cup D$), we claim below
	that $f$ and $P_U$ continue to be close\footnote{We are interchangeably using $f$ and $f^{(U)}$ here since the domain of the latter function can be treated as a domain of the former.} under this distribution for most $U$.
	\begin{claim}\label{clm:imbal-prob-bd}
		For all $U\supseteq D$ in ${[n]\choose 2k}$ and $\cD:=\cD_{U\setminus D,\change{i}}$ where $i\in [k-2d,\ell]$ is arbitrary, the following holds:
		$$\Pr_{\vecx\sim \cD}[f(\vecx) \ne P_U(\vecx)]\le 9^{d}\cdot \Pr_{\vecx\sim \bool^U_k}[f(\vecx) \ne P_U(\vecx)].$$
	\end{claim}

	Now, consider the joint probability distribution over $\change{T\in {[m]\choose t},}~ S_1,S_2\in {[m]\choose \ell}$
	and $\vecx \in {(S_1\cap S_2)\cup D\choose k}$
	on the LHS of~\eqref{eqn:agreement-1}. \change{Further let $U_1=S_1\cup D$ and $U_2=S_2\cup D$.}
	Note that the marginal distribution of $U_1$ (resp.~$U_2$) is uniform over $U\in {[n]\choose 2k}$ such that $U\supseteq D$. For each fixing of $U_1$, we note that the conditional distribution of $\vecx$ is given by \change{some convex combination of the distributions $(\cD_{U_1\setminus D,i})_{i}$}. The same also holds for $U_2$. Thus, using~\Cref{clm:imbal-prob-bd} for $U_1,U_2,\vecx$ sampled as above and~\eqref{eqn:delta-u-eta},
	we have that
	$$\Pr[P_{U_1}(\vecx) \ne P_{U_2}(\vecx)] \le \Pr[f(\vecx) \ne P_{U_1}(\vecx)] + \Pr[f(\vecx) \ne P_{U_2}(\vecx)] \le 2\cdot  9^d\cdot \eta.$$
	On the other hand,~\eqref{eqn:agreement-1}
	implies that
	$$\Pr[P_{U_1}(\vecx) \ne P_{U_2}(\vecx)] \ge \Omega_{d}(\gamma),$$
	which together imply that $\gamma\le O_{d}(\eta)$. \change{Therefore, using~\Cref{clm:poly}, we observe that the hypothesis of~\Cref{thm:agr-test} is satisfied for \(\gamma = O_{d}(\eta)\).}

	Hence, we get that the majority-decoded function
	$M:{[m]\choose \le d} \to G^{\leq d}[\bx|_{D}]$
	satisfies
	$$
		\Pr_{S\in {[m]\choose \ell}}\bigg[L_S\not\equiv M|_S\bigg]
		\le
		\exp(d^{O(1)})\cdot \gamma \le O_{d}(\eta).
	$$
	Similar to how \(P_{S \cup D}\) gave rise to the local function \(L_{S}\),
	we can use \(M\) to define a polynomial $P$ over variables $\vecx=(x_1,\dots, x_n)$
	by
	\begin{align}\label{eqn:poly-p-defn}
		P(\bx) := \sum_{C \in \binom{[m]}{\leq d}} M(C) \cdot \prod_{i\in C} x_i.
	\end{align}
	Since \(M(C)\) is the majority vote among \(L_{S}(C)\)
	which all have degree at most \(d - |C|\),
	we note that \(M(C)\) is also degree at most \(d - |C|\),
	which shows \(P\) indeed has degree at most \(d\).

	\paragraph{Showing that \(P\) and \(f\) are close.}
	Now, it remains to show that $P$ is actually close to $f$.
	For this, we first show that $P$ becomes identical to the local polynomials $P_U$ for many $U\in {[n]\choose 2k}$:
	\begin{claim}\label{clm:p-good-u}
		There exists $\eta'\le O_d(\eta)$ such that at least for a $1-\eta'$ fraction of subsets $U\in {[n] \choose 2k}$ such that $U\supseteq D$, it holds that $$P|_{U} \equiv P_U,$$ where $P|_U$ denotes the polynomial obtained by setting the variables in $\overline{U}$ to $0$ in $P$.
	\end{claim}

	Now, we use the above claim along with the fact that $f^{(U)}$ is close to $P_U$ for many $U$ to conclude that $f$ is close to $P$. Suppose for some $U\in {[n]\choose 2k}$ such that $U\supseteq D$, it holds that $P|_U \equiv P_U$ (By~\Cref{clm:p-good-u}, this occurs with probability at least $1-\eta'$). Then we obtain that $\Pr_{\vecx\in \bool^U_k}[f^{(U)}(\vecx) \ne P(\vecx)] = \Pr_{\vecx\in \bool^U_k}[f^{(U)}(\vecx) \ne P_U(\vecx)] = \delta_{U}$. Therefore, using~\eqref{eqn:delta-u-eta}, we get
	\begin{align}\label{eqn:prob-bound}
		\Pr_{\substack{U\in {[n]\choose 2k} \\ U\supseteq D\\ \vecx\in \bool^U_k}}[f(\vecx) \ne P(\vecx)] \le \eta' +  \E_{U\in {[n]\choose 2k} \text{~and~}U\supseteq D}\left[\delta_U\right] \le \eta' + \eta.
	\end{align}
	We are almost done, except for the fact that the marginal distribution of $\vecx$ on the LHS is not uniformly random over $\bool^{n}_k$. We remedy this below.

	We will treat $\vecx$ as a subset of $[n]$ instead of a Boolean vector.
	Then, we have for all $D'\subseteq D$ and $U\subseteq {[n]\choose 2k}$ such that $U\supseteq D$, the number of $\vecx \in {U\choose k}$ such that $\vecx \cap D = D'$ is at least
	$${|U\setminus D| \choose k-|D'|} \ge {|U\setminus D| \choose k} \ge \paren{1-\frac{k}{|U|-|D|}}^{|D|} \cdot {|U|\choose k} \ge \frac{1}{3^{2d}} \cdot {{|U| \choose k}} .$$ Hence, for every $D'\subseteq D$,~\eqref{eqn:prob-bound} implies that
	$$\Pr_{\substack{\vecx\in {[n]\choose k}\\\vecx \cap D = D'}}\bigg[f(\vecx) \ne P(\vecx)\bigg]\le 9^d\cdot (\eta + \eta') \le O_d(\eta),$$
	which when we take a union over all $D'\subseteq D$ yields the bound:
	$$\Pr_{\vecx \in \bool^n_k}[f(\vecx)\ne P(\vecx)] \le O_{d}(\eta).$$
	This finishes the proof of our key lemma~\Cref{lem:restrict}.
\end{proof}

It only remains to prove the three claims used in the proof of
the above lemma. We start with the proof of~\Cref{clm:poly}.

\begin{proof}[Proof of~\Cref{clm:poly}]
	Recall from~\eqref{eqn:pu} that for $i\in [2]$, we have
	$$P_{U_i} = \sum_{C \in \binom{S}{\leq d}} p_{i,C}(\bx|_{D}) \cdot \prod_{j\in C} x_j,$$ for some polynomials $p_{i,C}\in G^{\le d}[\vecx|_D]$ with the guarantee that all the monomials on the RHS (when expanded into a monomial representation) belong to the monomial basis $\cB_{U_i}$ for the domain $\bool^{U_i}_k$ given by~\Cref{prop:monomial-basis}. Observe that one can derive a polynomial representation of $P_{U_1}$ (resp.~$P_{U_2}$) restricted to the domain $\bool^{U_1 \cap U_2}_k$ by setting the variables in $S_1\setminus S_2$ (resp.~$S_2\setminus S_1$) to zero in the above representation. Let us call this restricted polynomial as $Q_1 \in G^{\le d}[\vecx|_{U_1\cap U_2}]$ (resp.~$Q_2$).

	 We show the forward direction of the equivalence first: If \(L_{S_1}|_{S_1\cap S_2} \equiv L_{S_2}|_{S_1\cap S_2}\), from our definition of $L_S$ using the above representation (see~\eqref{eqn:ls-defn}), we have that $p_{1,C} = p_{2,C}$ for all $C\in {S_1\cap S_2 \choose \le d}$. Therefore, \(P_{U_{1}}\) and \(P_{U_{2}}\) have the same polynomial
	representations when restricted to \(\bool^{U_{1} \cap U_{2}}_{k}\). That is $Q_1 = Q_2$ and so $P_{U_1} \equiv P_{U_{2}}$.

	Now we prove the backward direction of the equivalence: Suppose
	$L_{S_1}|_{S_1\cap S_2} \not\equiv L_{S_2}|_{S_1\cap S_2}$.
	Then there exists some \(C \in \binom{S_{1} \cap S_{2}}{\leq d}\) such that
	\begin{align*}
		p_{1,C}(\bx|_{D}) \neq p_{2,C}(\bx|_{D}),
	\end{align*}
	This implies that the basis representations of \(P_{U_{1}}\) and \(P_{U_{2}}\)
	restricted to the variables \(U_{1} \cap U_{2}\)
	are different, i.e,
	$$Q_1(\vecx|_{U_1\cap U_2})\ne Q_2(\vecx|_{U_1\cap U_2}).$$
	Now if we know that the polynomials $Q_1$ and $Q_2$ only contain monomials from a graded basis for
	\(\bool^{U_{1} \cap U_{2}}_{k}\),
	then this would imply they cannot represent the same function,
	and must be different functions on \(\bool^{U_{1} \cap U_{2}}_{k}\).
	Indeed, this is immediately implied by~\Cref{lem:restrn-basis} applied
	with \(n = |U_{1}|\) and \(S = S_{1} \setminus S_{2}\)\footnote{This is the reason we had to introduce
		the set \(D\), instead of applying \Cref{thm:agr-test} directly to the full set \([n]\).
		In that case we would not be guaranteed that different polynomial representations
		correspond to different functions on their overlap.} (resp.~$n=|U_2|$ and $S=S_2\setminus S_1$). Here we are using the fact that $Q_1$ (resp.~$Q_2$) only contains monomials from the basis $\cB_{U_1}$ (resp.~$\cB_{U_2}$) that do not contain variables from $S_1\setminus S_2$ (resp.~$S_2\setminus S_1$). We have thus shown that $Q_1$ and $Q_2$ are distinct functions over $\bool^{U_1\cap U_2}_k$, and so the same holds for $P_{U_1}$ and $P_{U_2}$ over this domain. That is, $P_{U_1} \not\equiv P_{U_2}$.
\end{proof}

We now prove~\Cref{clm:imbal-prob-bd}.

\begin{proof}[Proof of~\Cref{clm:imbal-prob-bd}]
	Let $\varepsilon' := \Pr_{\vecx \in {U\choose k}}[f(\vecx)\ne P_U(\vecx)]$.
	Consider the partition of the set ${U \choose k}$ corresponding to its intersection with $D$: For each $D'\subseteq D$, the number of $\vecx \in {U\choose k}$ such that $\vecx \cap D = D'$ is equal to
	$${|U\setminus D| \choose k-|D'|} \ge {|U\setminus D| \choose k} \ge \paren{1-\frac{k}{|U|-|D|}}^{|D|} \cdot {|U|\choose k} \ge \frac{1}{3^{2d}} \cdot {{|U| \choose k}},$$ using the facts that $|D'| \le |D| = 2d$ and $|U| = 2k \ge 8d$. Hence, \change{combined with the definition of $\varepsilon'$,} we have for each $D'\subseteq D$ that
	$$\Pr_{\substack{\vecx \in {U\choose k}\\ \vecx\cap D=D'}}\bigg[f(\vecx) \ne P_U(\vecx)\bigg] \le 3^{2d}\cdot \varepsilon'$$
	Now, conditioned on $\vecx\cap D = D'$ (for arbitrary $D'\subseteq D$), we observe that the conditional distribution of $\vecx$ for the priors being the uniform distribution over ${U\choose k}$ and $\cD_{U\setminus D\change{,i}}$ are identical. This follows from the fact that picking a random subset of size $k-|D'|$ of a random set in ${U\setminus D\choose i}$ produces a random set in ${U \setminus D\choose k-|D'|}$ (assuming $i\ge k-|D'|$). Hence, \change{by now taking a union over all $D'\subseteq D$ such that $i\ge k-|D'|$}, this finishes the proof of~\Cref{clm:imbal-prob-bd}.
\end{proof}

Finally we provide a proof of~\Cref{clm:p-good-u}.

\begin{proof}[Proof of~\Cref{clm:p-good-u}]
	Let $S \in {[n-2d] \choose k-2d}$ and let $U=S\cup D$. Recalling our definition of the polynomial $P$ from the majority-decoded function $M$ (i.e.,~\eqref{eqn:poly-p-defn}) we have that
	$$P(\bx) := \sum_{C \in \binom{[m]}{\leq d}} M(C) \cdot \prod_{i\in C} x_i.$$

	Now one can obtain $P|_U$ by setting the variables in $\overline{U}$ to 0 in the RHS of the above equation as the following:

	$$P|_U(\vecx|_U) = \sum_{C\in {S\choose \le d}} M(C) \cdot \prod_{i\in C} x_i.$$

	On the other hand, by the definition of $L_S$ (see~\eqref{eqn:ls-defn}) we have:
	$$P_U(\vecx|_U) = \sum_{C\in {S\choose \le d}} L_S(C)\cdot \prod_{i\in C} x_i.$$
	Now, applying the conclusion of~\Cref{thm:agr-test}, at least for a $1-\exp(d^{O(1)})\cdot \gamma = 1-O_d(\gamma)$ fraction of subsets $S\in {[n-2d]\choose k-2d}$, we have that
	\begin{align}\label{eqn:lsms}L_S \equiv M|_S.\end{align}
	Therefore, for all such $S$, using the above polynomial expressions, we obtain for $U=S\cup D$ that
	$$P|_U\equiv P_U,$$ which completes the proof of~\Cref{clm:p-good-u}.
\end{proof}

\section*{Acknowledgments}

We acknowledge the use of AI (specifically ChatGPT 5.4 Thinking) to guide us in the search for a proof of a version of~\Cref{lem:restrn-large-dist}.

\begin{sloppypar}
	\printbibliography[
		heading=bibintoc,
		title={References}
	] 
\end{sloppypar}

\appendix

\section{\texorpdfstring{Proofs from~\Cref{sec:prelims}}{Proofs from Preliminaries}}\label[appendix]{app:prelims}

We give a proof of~\Cref{lem:cart} giving a way to generate a graded basis for a Cartesian product of two domains.

\begin{proof}[Proof of~\Cref{lem:cart}]\label[appendix]{prf:cart}
	The set \(\mathcal{B}(\vecx,\vecy)\) is clearly downward-closed and monomial if both
	\(\mathcal{B}_{1}(\vecx), \mathcal{B}_{2}(\vecy)\) have these properties,
	so it remains to check that it
	is a monomial basis over \(D_{1} \times D_{2}\),
	and that it is graded if both \(\mathcal{B}_{1}(\vecx), \mathcal{B}_{2}(\vecy)\) are.

	Since the elements of \(\mathcal{B}_{1}(\bx)\) do not depend on \(D_{2}\),
	if a monomial \(a_{1}(\bx)\) has a representation
	\begin{align*}
		a_{1}(\bx) \equiv \sum_{m_{1} \in \mathcal{B}_{1}(\bx)} \alpha_{m_{1}} \cdot m_{1}
	\end{align*}
	as a function \(D_{1} \to G\), then it is also a representation
	as a function \(D_{1} \times D_{2} \to G\),
	and the same for a monomial \(a_{2}(\by)\)
	\begin{align*}
		a_{2}(\by) \equiv \sum_{m_{2} \in \mathcal{B}_{2}(\by)} \beta_{m_{2}} \cdot m_{2}.
	\end{align*}
	We then have that their product has the representation
	\begin{align*}
		a_{1}(\bx)a_{2}(\by)
		 & \equiv
		\left(\sum_{m_{1} \in \mathcal{B}_{1}(\bx)} \alpha_{m_{1}} \cdot m_{1}\right)
		\left( \sum_{m_{2} \in \mathcal{B}_{2}(\by)} \beta_{m_{2}} \cdot m_{2}\right) \\
		 & =
		\sum_{m_{1} \cdot m_{2} \in \mathcal{B}(\bx, \by)}
		\alpha_{m_{1}} \beta_{m_{2}} \cdot m_{1}m_{2}.
	\end{align*}
	This shows that every monomial lies in the
	span of \(\mathcal{B}(\bx,\by)\), and by linearity every polynomial
	function \(D_{1}\times D_{2}\to G\) lies in the span of
	\(\mathcal{B}(\bx,\by)\).

	We also need to show that any non-zero sum does not vanish on
	\(D_{1} \times D_{2}\).
	Let
	\begin{align*}
		\sum_{m_{1} \cdot m_{2} \in \mathcal{B}(\bx, \by)}
		c_{m_{1}m_{2}} \cdot m_{1}m_{2}
		= \sum_{m_{2} \in \mathcal{B}_{2}(\mathbf{y})} p_{m_{2}}(\mathbf{x}) \cdot
		m_{2}
	\end{align*}
	be a non-zero sum of the basis elements. Since it is non-zero,
	there exists at least one \(m_{2}\) such that \(p_{m_{2}}(\mathbf{x}) \neq 0\).
	Since \(\mathcal{B}_{1}(\mathbf{x})\) is a basis,
	there exists \(x_{0} \in D_{1}\) so that \(p_{m_{2}}(x_{0}) \neq 0\).
	Then we have
	\begin{align*}
		\sum_{m_{2} \in \mathcal{B}_{2}(\mathbf{y})} p_{m_{2}}(x_{0}) \cdot m_{2}
	\end{align*}
	is a non-zero sum of basis elements of \(\mathcal{B}_{2}(\mathbf{y})\),
	and so there exists \(y_{0} \in D_{2}\) so that the sum does not vanish on
	that point.
	It follows that the whole sum does not vanish at \((x_{0}, y_{0})\),
	which shows that \(\mathcal{B}(\mathbf{x}, \mathbf{y})\) is a basis.

	Lastly,
	if both \(\mathcal{B}_{1}\) and \(\mathcal{B}_{2}\) are graded,
	then the above proof shows that the product of the representations of
	$a_1(\vecx)$ and $a_2(\vecy)$ gives a representation of the correct degree,
	showing that \(\mathcal{B}(\bx, \by)\) is also graded.
\end{proof}

We now give a proof of the polynomial distance lemma (\Cref{lem:dist}) for a Cartesian product of a Boolean cube and slice.

\begin{proof}[Proof of~\Cref{lem:dist}]
	The proof is by an induction on $n_1$. The base case $n_1=0$ is proved by~{\cite[Lemma 5.1.6]{ABPSS25}} (for some absolute constant $\alpha>0$).

	Now, for $n_1\ge 1$, assume that the bound holds for $n_1-1$.
	Let $$P(\vecx,\vecy) = P_0(\vecx|_{[n_1-1]},\vecy) + x_{n_1} \cdot P_1(\vecx|_{[n_1-1]},\vecy),$$ for $(\vecx,\vecy)\in D$, where \(P_0\) and \(P_1\) are polynomials that do not contain the variable \(x_{n_1}\), and are of degree at most $d$ and $d-1$ respectively. If \(P_1|_{D} \equiv 0\), the statement follows by induction hypothesis since $P_0$ can be thought of as a function over $\bool^{n_1-1} \times \bool^n_{k}$. Otherwise \(P_1|_D \not\equiv 0\), in which case, applying the induction hypothesis (for degree parameter $d-1$), we obtain
	\begin{align}\label{eqn:induct-d-1}
		\Pr_{(x_1,\dots, x_{n_1-1},\by) \in \bg[n_{1}-1] \times \bool^{n}_k}[P_1(x_1,\dots, x_{n_1-1},\by) \neq 0]
		\geq \frac{1}{2^{d-1}}\cdot \paren{1-\frac{1}{n^\alpha}}.
	\end{align}
	Observe that whenever \(P_1(x_1,\dots, x_{n_1-1},\by) \neq 0\),
	the two values $$P(x_1,\dots, x_{n_1-1},0,\vecy) = P_0(x_1,\dots, x_{n_1-1},\vecy)$$ and $$P(x_1,\dots, x_{n_1-1},1,\vecy) = P_0(x_1,\dots, x_{n_1-1},\vecy)+ P_1(x_1,\dots, x_{n_1-1},\vecy)$$
	are distinct, so \(P(\vecx,\vecy) = 0\) for at most one of the two choices.
	Therefore, using~\eqref{eqn:induct-d-1}, we get
	\begin{align*}
		\Pr_{(\bx,\vecy) \in D}[P(\bx,\vecy) \neq 0]
		 & \geq
		\Pr_{(\vecx,\vecy)\in D}[P_1(\vecx|_{[n_1-1]},\by) \neq 0] \cdot \Pr_{(\vecx,\vecy)\in D}\bigg[P(\vecx,\vecy) \neq 0 ~\mid~ P_1(\vecx|_{[n_1-1]},\by) \neq 0\bigg] \\
		 & \geq \frac{1}{2^{d-1}}\cdot \paren{1-\frac{1}{n^\alpha}} \cdot \frac{1}{2}
		\ge \frac{1}{2^{d}}\cdot \paren{1-\frac{1}{n^\alpha}}.
	\end{align*}
\end{proof}

We also prove the bounds claimed in the ``consequently'' part of~\Cref{rem:sz-lem}.

\begin{proof}[Proof of~\Cref{rem:sz-lem}]
	Applying the polynomial distance lemma given by~{\cite[Lemma 5.1.6]{ABPSS25}} for the middle slice(s), for $n_1=0$, we directly have
	$$\frac{{n-2d \choose \lceil n/2 \rceil -d}}{{n\choose \lceil n/2 \rceil }} = \frac{{n-2d \choose \lfloor n/2 \rfloor -d}}{{n\choose \lfloor n/2 \rfloor }} = \frac{\lfloor n/2 \rfloor \cdots (\lfloor n/2 \rfloor-d+1)}{(n-1)(n-3)\dots (n-2d+1)} \cdot \frac{\lceil n/2 \rceil\dots(\lceil n/2 \rceil-d+1)}{n(n-2)\dots(n-2d+2)} \ge \frac{1}{2^{2d}}.$$
	Now, extending the bound from the slice to a Cartesian product of a cube and slice proceeds exactly like in the proof of~\Cref{lem:dist} that we gave just above this proof, except we change the induction hypothesis appropriately.
\end{proof}

\section{\texorpdfstring{Proofs from~\Cref{sec:small-dist}}{Proofs from Small Distance Section}}\label[appendix]{app:small-dist}

We justify here why the proof of~\Cref{lemma:main-informal} follows.

\begin{proof}[Proof of~\Cref{lemma:main-informal}]
	This proof is essentially the same as \cite[Lemma 3.2]{ABSS25-SZ-Lemma},
	where the indicator function \(\mathbbm{1}_{S}\) is replaced with a
	general function \(g\) with codomain \([0,1]\).
	The proof still holds, since the only property
	of \(\mathbbm{1}_{S}\) used is that
	\begin{align*}
		\parallel \mathbbm{1}_{S} \parallel_{p} \; = \rho^{1/p},
	\end{align*}
	but for general \(f\) we still have the upper bound
	\begin{align*}
		\parallel g \parallel_{p} \; \leq \rho^{1/p}.
	\end{align*}

\end{proof}

\section{\texorpdfstring{Proofs from~\Cref{sec:large-dist}}{Proofs from Large Distance Section}}

\subsection{Probability estimates}

Here, we provide a proof of~\Cref{clm:prob-estimates}.

\begin{proof}[Proof of~\Cref{clm:prob-estimates}]
	We start with estimating the first quantity. It is easy to see that
	\begin{align*}
		\Pr_{\mathbf{x} \in \bool^{2n}_{n}}[x_1 = x_2 = \dots = x_{t+1}] \; = \; 2 \cdot \Pr_{\mathbf{x} \in \bool^{2n}_{n}}[x_1 = x_2 = \dots = x_{t+1} = 0].
	\end{align*}
	So it suffices to focus on the event that $x_{1} , \cdots , x_{t+1}$ are all equal to $0$. The probability of this event occurring is equal to
	\begin{align*}
		\dfrac{\binom{2n-(t+1)}{n}}{\binom{2n}{n}}
		= \; \dfrac{\paren{n} \cdots \paren{n - t}}{(2n) \cdots \paren{2n-t}} \; \leq \; \paren{\frac{1}{2}}^{t+1}.
	\end{align*}
	Thus we get,
	\begin{align*}
		\Pr_{\mathbf{x} \in \bool^{2n}_{n}}[x_1 = x_2 = \dots = x_{t+1}] \; \leq \; 2 \cdot \dfrac{1}{2^{t+1}} \; = \; \dfrac{1}{2^{t}}.
	\end{align*}
	This finishes the probability estimate of the first event. Now we discuss the probability estimate of the second event. For the second event, setting $x_{1}, x_{3}, \ldots, x_{2t-1}$ also fixes the value of $x_{2}, x_{4}, \ldots, x_{2t}$. Varying over the number of $1$'s in the first $2t$ coordinates, we get,
	\begin{align*}
		\Pr_{\mathbf{x} \in \bool^{2n}_{n}}[x_{2i-1} = x_{2i} \; \text{ for all } \; i\in [t]] \; = \; \sum_{i = 0}^{t} \; \binom{t}{i} \cdot \dfrac{\binom{2n-2t}{n-2i}}{\binom{2n}{n}} \; \leq \; \; 2^{t} \cdot \max_{i=0}^t  \dfrac{\binom{2n-2t}{n-2i}}{\binom{2n}{n}}.
	\end{align*}
	Since $\binom{2n-2t}{n-2i}$ is always upper bounded by $\binom{2n-2t}{n-t}$, we get,
	\begin{align*}
		\Pr_{\mathbf{x} \in \bool^{2n}_{n}}[x_{2i-1} = x_{2i} \; \text{ for all } \; i\in [t]] \; \leq \; 2^{t} \cdot \dfrac{\binom{2n-2t}{n-t}}{\binom{2n}{n}}
	\end{align*}
	Expanding both numerator and denominator, we get an upper bound of
	$$2^t\cdot\paren{\frac{n-t+1}{2n-2t+1}}^{2t} \le \frac{1}{2^t}\cdot \paren{1+\frac{1}{n-t}}^{2t} \le \frac{1}{2^t}\cdot \exp(2t/(n-t)) \le \frac{1}{2^{t-1}},$$ where we are using $t\le n/4$ in the last step.
\end{proof}

\subsection{Global polynomial from pairwise agreements}\label[appendix]{app:large-final-agreement}

In this section, we reproduce the proof of~\cite{ASS} of~\Cref{clm:final-agreement} that guarantees the existence of a ``global'' degree-$d$ polynomial under an assumption of many ``pairwise consistent'' degree-$d$ polynomials.

\begin{proof}[Proof of~\Cref{clm:final-agreement}]
	The idea is to represent polynomials in a basis (not necessarily a monomial basis) such that one can ``read off'' the polynomial corresponding to the restrictions $z_{b_i} = 1-z_{a_i}$ from this representation.

	\paragraph{The star case:} In this case, we can assume that $(a_i,b_i)=(1,i+1)$ for $i\in [t]$, without loss of generality. For $\vecx = (x_1,\dots, x_n)$ denoting Boolean variables, abusing notation slightly, we also let $x_i : \bool^n \to \Z$ denote the function that outputs the $i$-th bit of its input for $i\in [n]$. We then note that the set of functions
	$$\cB = \bigg\{ \prod_{i\in S} x_i : S\subseteq [n] \bigg\}$$ forms a graded basis for $\bool^n$. We will now modify the above basis in a way that allows us to handle the restrictions of the form $z_{i+1} = 1-z_{1}$ for $i\in [t]$.

	For $i\in [n]$, define the function $y_i:\bool^n \to \Z$ as follows:
	\begin{align} \label{eqn:yis}
		y_i = \begin{cases}
			      x_{i} + x_1 - 1, & \text{ if }i\in [2,t+1], \\
			      x_i,             & \text{ otherwise.}
		      \end{cases}
	\end{align}
	We claim that the set of functions $$\cB' = \bigg\{ \prod_{i\in S} y_i : S\subseteq [n]\bigg\}$$ also forms a graded basis for $\bool^n$. To see that they form a graded basis, we will first show that, for $d'\le n$, each degree-$d'$ function of $\cB$ can be expressed as a linear combination of degree-$d'$ functions from $\cB'$: for $S\subseteq [n]$ of size $d'$, we have $$\prod_{i\in S} x_i = \prod_{i\in S\cap [2,t+1] }(y_i-y_1+1)\cdot \prod_{i\in S\cap ([n]\setminus [2,t+1])} y_i,$$ and by expanding out the RHS, we observe that indeed we get a linear combination of functions from $\cB'$ of degree at most $d'$. Furthermore, we note that the function $\prod_{i\in S} y_i$ appears only from the expansion of the monomial $\prod_{i\in S} x_i$. Thus, if we start with a function that is non-zero, since some function $\prod_{i\in S} x_i$ must have a non-zero coefficient (since $\cB$ is known to be a basis), we obtain that the corresponding coefficient (i.e., the one of $\prod_{i\in S} y_i$) must be non-zero. This completes the proof that $\cB'$ is a graded basis for $\bool^n$.

	Now, we represent the functions computed by $R_i$ and $R_{i,j}$ in the basis $\cB'$. In particular, let $(\alpha_{i,S})_{i\in [t],S\subseteq [n]} \in G$ be such that $$R_i(\vecx) = \sum_{S\subseteq [n]:|S|\le d} \alpha_{i,S}\cdot \prod_{i'\in S} y_{i'}.$$
	Using the above expression, we observe that if we set $z_{i+1} = 1-z_1$ in the polynomial $R_i(z_1,\dots, z_n)$, since all the terms containing $y_{i+1}$ (which is equivalent to $x_{i+1} + x_1 - 1$ by definition) vanish under this substitution, we get the following function for $i\in [t]$:
	$$R_{i}|_{z_{i+1}=1-z_1}(\vecx) = \sum_{S\subseteq [n]:|S|\le d\text{~and~}S\not\ni i+1} \alpha_{i,S}\cdot \prod_{i'\in S} y_{i'}.$$
	In fact, our choice of using the basis $\cB'$ instead of $\cB$ was motivated by the above simple way to express restricted functions. By the same reasoning, we have for $i\ne j\in [t]$:
	$$R_{i,j}(\vecx) = \sum_{S\subseteq [n]\setminus \{i+1,j+1\}:|S|\le d} \alpha_{i,S}\cdot \prod_{i'\in S} y_{i'}.$$
	Therefore, using the fact that $R_{i,j}(\vecx) = R_{j,i}(\vecx)$, the above expression yields that for $\vecx\in \bool^n$, we must have:
	$$\sum_{S\subseteq [n]\setminus \{i+1,j+1\}:|S|\le d} \alpha_{i,S}\cdot \prod_{i'\in S} y_{i'} = \sum_{S\subseteq [n]\setminus \{i+1,j+1\}:|S|\le d} \alpha_{j,S}\cdot \prod_{i'\in S} y_{i'}.$$
	However, since $\cB'=\{\prod_{i\in S} y_i:S\subseteq [n]\}$ is a basis for $\bool^n$, comparing coefficients on both sides, we deduce that
	$$\alpha_{i,S} = \alpha_{j,S}$$
	for all $i\ne j\in [t]$ and $S\subseteq [n]$ of size at most $d$ such that $S\subseteq [n]\setminus \{i+1,j+1\}$.

	Now, let $P(z_1,\dots, z_n)$ be the degree-$d$ polynomial corresponding to the function:
	$$P(\vecx) = \sum_{S\subseteq [n]:|S|\le d} \beta_{S} \cdot \prod_{i'\in S} y_{i'},$$
	where $\beta_S = \alpha_{i,S}$
	for arbitrary indices $i\ne j\in [t]$ such that
	$S\subseteq [n]\setminus \{i+1,j+1\}$;
	since $|S|\le d \le t-2$, such indices always exist.
	It is important to notice here that the exact choice of $i$ and $j$ does not matter in that the resulting function $P$ is always the same due to the condition that $\alpha_{i,S} = \alpha_{j,S}$ whenever $S\subseteq [n]\setminus\{i+1,j+1\}$.

	For this function $P$, we will now show that we get the same function under the substitution $z_{i+1} = 1-z_{1}$. We have for all $i\in [t]$ that: $$P|_{z_{i+1}=1-z_1}(\vecx) = \sum_{S\subseteq [n]:|S|\le d\text{~and~}S\not\ni i+1} \beta_{S}\cdot \prod_{i'\in S} y_{i'},$$ and $$R_i|_{z_{i+1}=1-z_1}(\vecx) = \sum_{S\subseteq [n]:|S|\le d\text{~and~}S\not\ni i+1} \alpha_{i,S}\cdot \prod_{i'\in S} y_{i'}.$$
	For every $S\subseteq[n]$ of size at most $d$ such that $S\not\ni i+1$, we recall that $\beta_S$ (almost by definition) is equal to $\alpha_{i,S}$. Thus, since the terms in the RHS of the above equations come from a basis (namely $\cB'$), we conclude that indeed $P|_{z_{i+1}=1-z_1} = R_i|_{z_{i+1}=1-z_1}$ as polynomials. Hence $P$ is the desired polynomial, finishing the proof of~\Cref{clm:final-agreement} in the star case.

	\paragraph{The matching case:} The argument is quite similar to the star case. First, without loss of generality, we assume that the matching is formed by the edges $(a_i,b_i) = (2i-1,2i)$ for $i\in [t]$. We now define the basis $\cB'$ different from the star case as follows: For $i\in [n]$, define the function $y_i:\bool^n \to \Z$ as
	\begin{align} \label{eqn:yis-matching} y_i =
		\begin{cases}
			x_{i} + x_{i-1} - 1, & \text{~if~}i\in \{2,4,\dots, 2t\}, \\
			x_i,                 & \text{otherwise.}
		\end{cases}
	\end{align} and the family $\cB'$ as
	$$\cB' =\bigg\{\prod_{i\in S} y_i : S\subseteq [n] \bigg\}.$$
	Similar to the star case, it can be argued that $\cB'$ is a graded basis for $\bool^n$.

	Now, let $(\alpha_{i,S})_{i\in [t], S\subseteq [n]}\in G$ be such that
	$$R_i(\vecx) = \sum_{S\subseteq [n]:|S|\le d} \alpha_{i,S}\cdot \prod_{i'\in S} y_{i'}.$$ Using the above expression, we have the following expression for the polynomial obtained on substituting $z_{2i} = 1-z_{2i-1}$ where $i\in [t]$:
	$$R_i|_{z_{2i}=1-z_{2i-1}}(\vecx) = \sum_{S\subseteq [n]:|S|\le d \text{~and~}S\not\ni 2i} \alpha_{i,S}\cdot \prod_{i'\in S} y_{i'}.$$ Similarly, for $i\ne j\in [t]$, we have:
	$$R_{i,j}(\vecx) = \sum_{S\subseteq [n]\setminus \{2i,2j\}:|S|\le d} \alpha_{i,S} \cdot \prod_{i'\in S} y_{i'}.$$ Therefore, using the fact that $R_{i,j}(\vecx) = R_{j,i}(\vecx)$, we get for $\vecx\in \bool^n$ that
	$$\sum_{S\subseteq [n]\setminus \{2i,2j\}:|S|\le d} \alpha_{i,S} \cdot \prod_{i'\in S} y_{i'} = \sum_{S\subseteq [n]\setminus \{2i,2j\}:|S|\le d} \alpha_{j,S} \cdot \prod_{i'\in S} y_{i'}.$$ Thus, since the RHS terms are from the basis $\cB'$, we conclude for all $i\ne j\in [t]$ and $S\subseteq [n]\setminus \{2i,2j\}$ of size at most $d$, it holds that $$\alpha_{i,S} = \alpha_{j,S}.$$

	Now, let $P(z_1,\dots, z_n)$ be the degree-$d$ polynomial corresponding to the function:
	$$P(\vecx) = \sum_{S\subseteq [n]:|S|\le d} \beta_{S} \cdot \prod_{i'\in S} y_{i'},$$ where $\beta_S = \alpha_{i,S}$ for arbitrary indices $i\ne j\in [t]$ such that $S\subseteq [n]\setminus \{2i,2j\}$; since $|S|\le d \le t-2$, such indices always exist. It is important to notice here that the exact choice of $i$ and $j$ does not matter in that the resulting function $P$ is always the same due to the condition that $\alpha_{i,S} = \alpha_{j,S}$ whenever $S\subseteq [n]\setminus\{2i,2j\}$.

	For this function $P$, we will now show that we get the same function under the substitution $z_{2i} = 1-z_{2i-1}$. We have for all $i\in [t]$ that: $$P|_{z_{2i}=1-z_{2i-1}}(\vecx) = \sum_{S\subseteq [n]:|S|\le d\text{~and~}S\not\ni 2i} \beta_{S}\cdot \prod_{i'\in S} y_{i'},$$ and $$R_i|_{z_{2i}=1-z_{2i-1}}(\vecx) = \sum_{S\subseteq [n]:|S|\le d\text{~and~}S\not\ni 2i} \alpha_{i,S}\cdot \prod_{i'\in S} y_{i'}.$$
	For every $S\subseteq[n]$ of size at most $d$ such that $S\not\ni 2i$, we recall that $\beta_S$ (by definition) is equal to $\alpha_{i,S}$. Thus, since the terms in the RHS of the above equations come from a basis (namely $\cB'$), we conclude that indeed $P|_{z_{2i}=1-z_{2i-1}} = R_i|_{z_{2i}=1-z_{2i-1}}$ as polynomials. Hence $P$ is the desired polynomial, finishing the proof of~\Cref{clm:final-agreement} in the matching case as well.
\end{proof}

\end{document}